\documentclass[preprint,12pt]{elsarticle}

\usepackage{amsmath,amssymb,amsthm,amsfonts,latexsym,graphicx,subfigure}
\usepackage{mathrsfs}
\usepackage[dvips]{color}
\usepackage{indentfirst}
\usepackage{fancyhdr}
\usepackage{bbm}
\usepackage{accents}

\biboptions{sort&compress}

\makeatletter
\def\wbar{\accentset{{\cc@style\underline{\mskip8mu}}}}

\makeatother

\def\pd#1#2{\frac{\partial #1}{\partial #2}}
\renewcommand{\vec}[1]{\mbox{\boldmath \small $#1$}}

\def\mi{\mathtt{i}}
\def\me{\mathrm{e}} 
\newcommand{\pp}[2]{\frac{\partial{#1}}{\partial{#2}}}
\newcommand{\dd}[2]{\frac{\dif{#1}}{\dif{#2}}}

\def\disp{\displaystyle}
\def\dif{\mathrm{d}}

\newcommand{\diag}{\mathrm{diag}}
\newcommand{\imag}{\mathrm{Im}}
\newcommand{\real}{\mathrm{Re}}

\newtheorem{Theorem}{Theorem}
\newtheorem{Proposition}[Theorem]{Proposition}
\numberwithin{Theorem}{section}

\newtheorem{Lemma}[Theorem]{Lemma}
\newtheorem{Remark}{Remark}
\numberwithin{Remark}{section}
\theoremstyle{definition}
\newtheorem{Example}{Example}[section]
\numberwithin{equation}{section}

\newcommand{\videpost}{\textit{vide post}{}}
\newcommand{\etc}{\textit{etc}{}}

\newcommand{\ie}{\textit{i.e.}{~}}
\newcommand{\eg}{\textit{e.g.}{~}}

\newcommand{\vgamma}{\vec{\gamma}}
\newcommand{\vGamma}{\vec{\Gamma}}
\newcommand{\vPsi}{\vec{\Psi}}
\newcommand{\vPhi}{\vec{\Phi}}
\newcommand{\vpsi}{\vec{\psi}}

\newcommand{\vK}{\vec{K}}
\newcommand{\vE}{{E}}
\newcommand{\vI}{\vec{I}}
\newcommand{\vdelta}{{\delta}}
\newcommand{\vell}{{\ell}}

\journal{Journal of Computational Physics}

\begin{document}

\begin{frontmatter}

\title{Numerical methods for nonlinear Dirac equation}
\author[ss]{Jian Xu}
\ead{xujian764@gmail.com}

\author[ss]{Sihong Shao\corref{cor}}
\ead{sihong@math.pku.edu.cn}

\author[thz]{Huazhong Tang}
\ead{hztang@math.pku.edu.cn}

\cortext[cor]{Corresponding author.}

\address[ss]{LMAM and School of Mathematical Sciences, Peking University,
Beijing 100871, China}
\address[thz]{HEDPS, CAPT \& LMAM, School of Mathematical Sciences, Peking University,
Beijing 100871, China}
\begin{abstract}
This paper presents a review of the current state-of-the-art of numerical methods for
nonlinear Dirac (NLD) equation.
Several methods are extendedly proposed for the (1+1)-dimensional NLD equation
with the scalar and vector self-interaction and analyzed in the way of the
accuracy and the time reversibility
as well as the conservation of the discrete charge, energy and linear momentum.
Those methods are the Crank-Nicolson (CN) schemes,
the linearized CN schemes, the odd-even hopscotch scheme, the leapfrog scheme,
a semi-implicit finite difference scheme, and the exponential operator splitting (OS) schemes.
The nonlinear subproblems resulted from the OS schemes are analytically solved
by fully exploiting the local conservation laws of the NLD equation.
The effectiveness of the various numerical methods,
with special focus on the error growth and the computational cost,
is illustrated on two numerical experiments,
compared to two high-order accurate Runge-Kutta discontinuous Galerkin methods.
Theoretical and numerical comparisons show that the high-order accurate OS schemes
may compete well with other numerical schemes discussed here in terms of the accuracy
and the efficiency.
A fourth-order accurate OS scheme is further applied to investigating the interaction dynamics
of the NLD solitary waves under the scalar and vector self-interaction.
The results show that the interaction dynamics
of two NLD solitary waves   depend on the exponent power of the self-interaction in the NLD equation;
collapse happens after collision of two equal one-humped NLD solitary waves
under the cubic vector self-interaction in contrast to no collapse scattering for corresponding quadric case.
\end{abstract}
\end{frontmatter}

\section{Introduction}
\label{sec:intro}

As a relativistic wave equation,
the Dirac equation provides naturally a description of an electron \cite{Dirac1928}.
Following Dirac's discovery of the linear equation of the electron,
there appears the fundamental idea of nonlinear description of an elementary spin-1/2 particle
which makes it possible to take into account its self-interaction \cite{Ivanenko1938,FinkelsteinLelevierRuderman1951,FinkelsteinFronsdalKaus1956}.
Heisenberg put forward
the idea to use a nonlinear Dirac (NLD) equation
as a possible basis model for a unified field theory \cite{Heisenberg1957}.
A key feature of the NLD equation is that
it allows solitary wave solutions or particle-like solutions --- the stable localized solutions
with finite energy and charge \cite{Ranada1983}.
That is, the particles appear as intense localized regions of field
which can be recognized as the basic ingredient
in the description of extended objects in quantum field theory \cite{Weyl1950}.
Different self-interactions give rise to different NLD models
mainly including
the Thirring model \cite{Thirring1958},
the Soler model \cite{Soler1970},
the Gross-Neveu model \cite{GrossNeveu1974}
(equivalent to the massless Soler model),
and the bag model \cite{Mathieu1984}
(\ie the solitary waves with finite compact support),
all of which attracted wide interest of physicists
and mathematicians around the 1970s and 1980s,
especially on looking for the solitary wave solutions
and investigating the related physical and mathematical properties \cite{Ranada1983}.

For the NLD equation in (1+1) dimensions (\ie one time dimension plus one space dimension),
several analytical solitary wave solutions are derived in \cite{LeeKuoGavrielides1975,ChangEllisLee1975}
for the quadric nonlinearity, \cite{Mathieu1985-prd} for fractional nonlinearity as well as
\cite{Stubbe1986-jmp,CooperKhareMihailaSaxena2010} for general nonlinearity
by using explicitly the constraint resulting from energy-momentum conservation,
and summarized by Mathieu \cite{Mathieu1985-jpa-mg}. In contrast,
there are few explicit solutions in (1+3) dimensions except for some
particular cases shown in \cite{Werle1977-lmp} in spite of
their existence claimed by mathematicians for various situations
\cite{Vazquez1977,Werle1981-appb,MathieuMorris1984,CazenaveVazquez1986,Merle1988-jde,BalabaneCazenaveDouadyMerle1988,BalabaneCazenaveVazquez1990,EstebanSere1995}
(the readers are referred to an overview \cite{EstebanSere2002} on this topic),
and most understanding is based on
numerical investigations, \eg \cite{Rafelski1977-plb,Takahashi1979-jmp,Alvarez1985}.
Beyond this, the study of the NLD equation in (1+1) dimensions could be very helpful for that in (1+3) dimensions
since the (1+1)-dimensional NLD equation correspond to the asymptotic form of the equation in
the physically interesting case of (1+3) dimensions as emphasized by Kaus \cite{Kaus1976}.
That is, some qualitative properties of the NLD solitary waves could be similar in such two cases.
An interesting topic for the NLD equation is the stability issue, which has been the central topic in works
spread out over several decades in an effort that is still ongoing.
Analytical studies of the NLD solitary wave stability face serious obstacles \cite{StraussVazquez1986,AlvarezSoler1986,BlanchardStubbeVazquez1987},
while results of computer simulations are contradictory \cite{Bogolubsky1979-pla,AlvarezSoler1983,Mathieu1983,Alvarez1985}.
The stability analysis of the NLD solitary waves is still a very challenging mathematical problem to date \cite{Barashenkov1998,CooperKhareMihailaSaxena2010}.
Recent efforts in this direction can be found in \cite{Chugunova2006,Pelinovsky2010,Haddad2011-epl,Comech2011,Boussaid2012,Berkolaiko2012}.
Another rising mathematical interest related to the NLD equation is the analysis of global well-posedness,
\eg see \cite{Pelinovsky2010-arxiv,Bournaveas2012} and references therein.

In the case of that
theoretical methods were not able to give the satisfactory results,
numerical methods were used for obtaining
the solitary wave solutions of the NLD equation
as well as for investigating the stability.
An important step in this direction was made by
Alvarez and Carreras in 1981 \cite{AlvarezCarreras1981},
who simulated the interaction dynamics
between the (1+1)-dimensional NLD solitary waves of
different initial charge for the Soler model \cite{Soler1970} by
using a second-order accurate Crank-Nicholson (CN) scheme \cite{AlvarezKuoVazquez1983}.
They first saw there: charge and energy interchange except for some particular
initial velocities of the solitary waves; inelastic interaction in binary collisions;
and oscillating state production from binary collisions.
Motivated by their work, Shao and Tang revisited this interaction dynamics problem in 2005 \cite{ShaoTang2005}
by employing a fourth-order accurate Runge-Kutta discontinuous Galerkin (RKDG) method \cite{ShaoTang2006}.
They not only recovered the phenomena mentioned by Alvarez and Carreras but also revealed
several new ones, \eg
collapse in binary and ternary collisions of two-humped NLD solitary waves \cite{ShaoTang2005};
a long-lived oscillating state formed with an approximate constant frequency in collisions of two standing waves \cite{ShaoTang2006};
full repulsion in binary and ternary collisions of out-of-phase waves \cite{ShaoTang2008}.
Their numerical results also inferred that the two-humped profile could undermine
the stability during the scattering of the NLD solitary waves.
Note in passing that the two-humped profile was first pointed out by Shao and Tang \cite{ShaoTang2005}
and later gotten noticed by other researchers \cite{CooperKhareMihailaSaxena2010}.
Besides the often-used CN \cite{AlvarezKuoVazquez1983}
and RKDG methods \cite{ShaoTang2006},
there exist many other numerical schemes for solving the (1+1)-dimensional
NLD equation:
split-step spectral schemes \cite{FrutosSanz-Serna1989},
the linearized CN scheme \cite{Alvarez1992},
the semi-implicit scheme \cite{Jimcnez1994}\cite{Bournaveas2012},
Legendre rational spectral methods \cite{WangGuo2004},
multi-symplectic Runge-Kutta methods \cite{HongLi2006},
adaptive mesh methods \cite{WangTang2007} \etc.
The fourth-order accurate RKDG method \cite{ShaoTang2006}
is very appropriate for investigating the interaction dynamics of the NLD solitary waves
due to their ability to capture the discontinuous or strong gradients without producing spurious oscillations,
and thus performs better than the second-order accurate CN scheme \cite{AlvarezKuoVazquez1983}.
However, the high cost due to the relatively more freedoms used in each cell
and the stringent time step constraint
reduce its practicality
in more realistic simulations where realtime and quantitative results are required.

Recently, there has been a remarkable upsurge
of the interest in the NLD models,
as they emerge naturally as practical models
in other physical systems,
such as the gap solitons in nonlinear optics \cite{Barashenkov1998},
Bose-Einstein condensates in honeycomb optical lattices \cite{Haddad2009},
as well as matter influencing
the evolution of the Universe in cosmology \cite{Saha2012}.
In view of such new trend,
longtime stable, efficient, conservative and high-order accurate numerical methods
for solving the NLD equation are highly desirable.
Finite difference methods, usually as the first try in practice,
enable easy coding and debugging and thus are often used by physicists and engineers.
However, detailed discussion and careful comparison on finite difference solvers for the NLD equation
are not existed.
To this end, the present work as the first step will extendedly propose
the finite deference schemes for solving the NLD equation with the scalar and vector self-interaction.
A general and precise comparison among them will be presented.
However, all of these finite difference methods are often of the second order accuracy and thus
sustain fast error growth with respect to time. To achieve relatively slow error growth,
high-order accurate numerical methods are required.
By exploiting the local conservation laws of the NLD equation,
we present exponential operator splitting (OS) schemes which are time reversible and
can conserve the discrete charge. One of the high-order accurate OS schemes is afterwards
adopted to investigate the interaction dynamics for
the NLD solitary waves under the general scalar and vector self-interaction.
It should be noted that
the experiments carried out in the literatures
are all limited to the collisions of
the NLD solitary waves under the quadric scalar self-interaction.
Here, the binary collisions of the NLD solitary waves under the cubic scalar self-interaction
 or under the vector self-interaction or under the linear combination of the scalar and vector self-interactions
are all studied for the first time.

The paper is organized as follows. There is a brief review of the NLD equation in
Section \ref{sec2:gnld} and the solitary wave solutions  are also derived there for
the general scalar and vector self-interaction.
The numerical schemes are presented in Section \ref{sec:DiffSch} and
corresponding numerical analysis is given
in Section \ref{sec:discussion4fdm}. The numerical results are presented with discussion in Section \ref{sec:discussion4fdm}.
The paper is concluded in Section \ref{sec:conclusion} with a few remarks.

\section{Preliminaries and notations}
\label{sec2:gnld}
This section will introduce the (1+1)-dimensional NLD equation
with the scalar and vector self-interaction and derive its two solitary wave solutions.
Throughout the paper, units in which both the speed of light and the reduced Planck constant
are equal to one will be used.

\subsection{Nonlinear Dirac equation}

Our attention is restricted to the NLD equation in $(1+1)$ dimensions
which can be written in the covariant form
\begin{equation}
(\mi\vgamma^\mu\partial_\mu-m)\vPsi+\partial L_\text{I}[\vPsi]/\partial \wbar{\vPsi}=0, \label{generalnld}
\end{equation}
where $\vPsi$ is the spinor with  two complex components,
$\wbar{\vPsi}:=\vPsi^\dag\vgamma^0$ denotes the adjoint spinor,
$\vPsi^\dag$ is the  complex conjugate transpose of $\vPsi$,
$L_\text{I}[\vPsi]$ denotes the self-interaction Lagrangian,
the Greek index $\mu$  runs from 0 to 1,
the Einstein summation convection has been applied,
$\mi$ is the imaginary unit,
$m$  is the rest mass,
$\partial_\mu=\pp{}{x^\mu}$ stands for the covariant derivative,
and $\vgamma^\mu$, for $\mu=0,1$, are the gamma matrices or the Dirac matrices, chosen
as those in \cite{AlvarezCarreras1981,ShaoTang2005},
$$
\vgamma^0=\begin{pmatrix} 1& 0\\ 0 &-1\end{pmatrix},\quad
\vgamma^1=\begin{pmatrix} 0& 1\\ -1 &0\end{pmatrix}.
$$
In fact, Eq. \eqref{generalnld} is the equation of motion for the classical spinorial particle
with the Lagrangian  being a sum of  the Dirac Lagrangian and the self-interaction Lagrangian, \ie
\begin{equation}
L[\vPsi]=\wbar{\vPsi} (\mi
\vgamma^\mu\partial_\mu-m)\vPsi+{L}_\text{I}[\vec\Psi]. \label{total_lag}
\end{equation}
There exist several different NLD models in the literature,
where two important models in (1+1) dimensions are the scalar self-interaction of Soler \cite{Soler1970}
\begin{equation}\label{self-scalar1}
L_{\text{s}}[\vPsi]:=\wbar{\vPsi}\vPsi,
\end{equation}
and the vector self-interaction of Thirring \cite{Coleman1975}
\begin{equation}\label{self-vector1}
L_{\text{v}}[\vPsi]:= \wbar{\vPsi}\vgamma^\mu\vPsi\wbar{\vPsi}\vgamma_\mu\vPsi,
\end{equation}
where $\vgamma_\mu=\eta_{\mu\nu}\vgamma^\nu$ with the Minkowski metric
$\eta_{\mu\nu}=\text{diag}(1,-1)$ on spacetime,
which implies $\vgamma_\mu=(-1)^\mu\vgamma^\mu$.

This paper will focus on the NLD equation \eqref{generalnld}
with a more general self-interaction \cite{Stubbe1986-jmp,NogamiToyama1992}
\begin{equation}\label{generalLI}
{L}_\text{I}[\vPsi] = s \left(L_{\text{s}}[\vPsi]\right)^{k+1}
        + v \left(L_{\text{v}}[\vPsi]\right)^{(k+1)/2},
\end{equation}
and extendedly propose and compare its numerical methods,
where $s$ and $v$ are two real numbers and $k$ is a positive real number.
If the spinor $\vPsi$ is scaled by a constant factor as $\vPsi^\prime=\sqrt{\alpha}\vPsi$
with $\alpha\in\mathbb{C}$,
then the scaled self-interaction Lagrangian will be
$\alpha^{k+1}L_\text{I}[\vPsi]$ which shows that the power exponent to $\alpha$ is $k+1$.
In such sense, we call that the self-interaction Lagrangian $L_\text{I}$ has the power exponent $k+1$
\cite{Mathieu1985-prd,Stubbe1986-jmp,CooperKhareMihailaSaxena2010}.
Hereafter the \textit{quadric} and \textit{cubic} self-interaction will be referred to the case $k=1$
and the case $k=2$, respectively.

The  self-interaction \eqref{generalLI} implies
the so-called homogeneity relation \cite{Mathieu1985-jpa-mg,Stubbe1986-jmp}
\begin{equation*}
\wbar{\vPsi}\pp{ {L}_\text{I}[\vPsi]}{\wbar{\vPsi}}=(k+1){L}_\text{I}[\vPsi]. \label{ppL=L}
\end{equation*}
Combining it with the definition of the Lagrangian $L[\vPsi]$ and \eqref{generalnld}
gives
\begin{equation} \label{LLIrelation}
 {L}[\vPsi] = -k  {L}_\text{I}[\vPsi].
\end{equation}

For the NLD equation \eqref{generalnld} with  \eqref{generalLI},
one may still verify the following local conservation laws for the current vector $J_\mu$
and the energy-momentum tensor $T_{\mu\nu}$:
\begin{equation}
\partial^\mu J_\mu  =0, \quad
\partial^\mu T_{\mu\nu}  =0, \label{Tconserve}
\end{equation}
where
\begin{equation*}
J_\mu  = \wbar{\vPsi}\vgamma_\mu\vPsi,\quad
T_{\mu\nu} =\mi\wbar{\vPsi}\vgamma_\mu\partial_\nu\vPsi-\eta_{\mu\nu}{L}[\vPsi].
\end{equation*}

For localized solitary waves $\vPsi=(\psi_1,\psi_2)^T$,
one may derive a direct corollary of \eqref{Tconserve},
\ie the following global conservation laws \cite{ShaoTang2006}.
\begin{Proposition}\label{prop-cl}
Assume that $\disp \lim_{|x|\rightarrow +\infty}|\psi_i(x,t)|=0$ and
$\disp \lim_{|x|\rightarrow +\infty}|\partial_x{\psi_i}(x,t)|<+\infty$
hold uniformly for $t\geq 0$ and $i=1,2$. The total energy $E$, the total linear momentum
$P$, and the total charge $Q$, defined  respectively by
\begin{equation}\label{qep-infty}
E(t):= \int_{-\infty}^{+\infty} T_{00} \dif x, \quad
P(t):= \int_{-\infty}^{+\infty} T_{01} \dif x, \quad
Q(t):= \int_{-\infty}^{+\infty} J_0 \dif x,
\end{equation}
satisfy
$$ \frac{\dif
}{\dif t}E(t)=0, \quad \frac{\dif}{\dif t}P(t)=0, \quad \frac{\dif}{\dif t}Q(t)=0.$$
\end{Proposition}

The properties \eqref {LLIrelation} and \eqref{Tconserve} will be also exploited to find the  solitary wave solutions of the
 (1+1)-dimensional NLD equation \eqref{generalnld} with $L_\text{I}$ given in \eqref{generalLI}
 in the next subsection.

\subsection{Standing wave solution}
\label{appA:derivation}
This subsection will derive the standing wave solutions of
the (1+1)-dimensional NLD equation \eqref{generalnld}
with the self-interaction \eqref{generalLI} in the spirit of the technique used in \cite{LeeKuoGavrielides1975,ChangEllisLee1975,Mathieu1985-jpa-mg,Stubbe1986-jmp}.
The solution $\vPsi=(\psi_1,\psi_2)^T$
of the  NLD equation \eqref{generalnld} with $L_\text{I}$ in \eqref{generalLI},
having the form
\begin{equation*}
\psi_1(x,t)=\me^{-\mi \omega t}\varphi(x),\quad \psi_2(x,t)=\me^{-\mi \omega t}\chi(x),
\end{equation*}
is wanted, where $m>\omega\geq 0$,
and $|\varphi(x)|$ and $|\chi(x)|$ are assumed to decay exponentially as $|x|\rightarrow +\infty$ or
have the finite compact support.
For such  solution, it is not difficult to verify
that both the Lagrangian $L[\vPsi]$ and the energy-momentum tensor $T_{\mu\nu}$
are independent of the time $t$, because
\begin{equation}\label{total_lag_ti}
\begin{aligned}
{L}[\vPsi] &\equiv \omega \wbar{\vpsi}\vgamma^0\vpsi+\mi\wbar{\vpsi}\vgamma^1\vpsi_x-m\wbar{\vpsi}\vpsi
+{L}_\text{I}[\vPsi], %
\\
T_{00} &\equiv-\mi \wbar{\vpsi}\vgamma^1\vpsi_x+m\wbar{\vpsi}\vpsi- {L}_\text{I}[\vPsi],
\quad T_{01}\equiv \mi\wbar{\vpsi}\vgamma^0\vpsi_x,  \\
T_{10} &\equiv-\omega \wbar{\vpsi}\vgamma^1\vpsi, \quad T_{11}\equiv -\mi \wbar{\vpsi}\vgamma^1\vpsi_x+{L}[\vPsi],
\end{aligned}
\end{equation}
where $\vec \psi(x)=\big(\varphi(x), \chi(x)\big)^T$.
Using the conservation laws \eqref{Tconserve} further  gives
\begin{equation}
T_{10}=-\omega \wbar{\vpsi}\vgamma^1\vpsi =0, \quad 
T_{11}=-\mi\wbar{\vpsi}\vgamma^1\vpsi_x+ {L}[\vPsi] =0. \label{t11}
\end{equation}
The first equation implies that $\varphi^\ast\chi$ is imaginary  because
\begin{equation*}
\wbar{\vpsi}\vgamma^1\vpsi
=\varphi^\ast\chi+\varphi\chi^\ast=0.
\end{equation*}
Thus, without loss of generality, we may assume that $\varphi(x)$ is real and $\chi(x)$ is imaginary,
and they are in the form
\begin{equation}
\vpsi(x)=\left(
\begin{matrix}
\varphi(x)\\
\chi(x)
\end{matrix}
\right)=
R(x) \left( \begin{matrix}
  \cos\big(\theta(x)\big) \\
  \mi \sin\big(\theta(x)\big)\end{matrix}\right),\label{phaseexp}
\end{equation}
where both $R(x)$ and $\theta(x)$ are pending real functions, and $R(x)$ is
assumed to satisfy the inequality $Q(t)\equiv \int_{-\infty}^{+\infty} R^2(x)\dif x<+\infty$.
On the other hand,  combining the first equation in \eqref{total_lag_ti} with the second equation in \eqref{t11}
yields
\begin{equation}\label{key-Li}
\omega \vpsi^\dag\vpsi- m\wbar{\vpsi}\vpsi +  L_\text{I}[\vPsi]=0,
\end{equation}
which becomes for \eqref{phaseexp}
\begin{equation}
  \label{eq:Li-2}
   R^2 \big(\omega- m  \cos(2\theta)\big) + R^{2k+2}\big(s\cos^{k+1}(2\theta) + v\big) = 0.
\end{equation}
Combining \eqref{key-Li} with \eqref{LLIrelation} leads to
\begin{equation*}
k \omega \vpsi^\dag\vpsi-k m\wbar{\vpsi}\vpsi-\mi\wbar{\vpsi}\vgamma^1\vpsi_x=0,
\label{key-eq1}
\end{equation*}
which reduces to for \eqref{phaseexp}
\begin{equation} \label{theta}
\dd{\theta}{x}=-k
\big(\omega - m\cos (2\theta)\big).
\end{equation}
Because
\begin{equation*}
\int_{0}^{\theta} \frac{\dif \tilde{\theta}}{-k\omega+km \cos (2 \tilde{\theta})}
={\frac{1}{
   \sqrt{k^2(m^2-\omega^2)}} \tanh
   ^{-1}\left(\frac{k(m+\omega)}{\sqrt{k^2(m^2-\omega^2)}}\tan
   \left(\theta\right)\right)},
\end{equation*}
for $\theta \in \big(-0.5 \cos^{-1}(\omega/m), 0.5 \cos^{-1}(\omega/m)\big)\subset(-\pi/2,\pi/2)$ when $m>\omega\geq 0$,
the solution of \eqref{theta} may be derived  as follows:
\begin{equation}
\theta(x) = \tan^{-1} \left(  \sqrt{\frac{m-\omega}{m+\omega}}  \tanh\Big(k  \sqrt{m^2-\omega^2} x \Big)\right),
\label{theta-1}
\end{equation}
for initial data $\theta(0) =0$ and $k>0$. In fact,
under the previous assumption, one may verify  $\theta(x)\in (-\pi/4,\pi/4)$.
{If coefficients $s$ and $v$ in \eqref{generalLI}  belong to the set
$\{v\geq 0, s>-v\}$} for $m>\omega>0$, or $\{v>0, s>-v\}$ for $m>\omega=0$,
 then from Eq. \eqref{eq:Li-2} one has  non-trivial  $R(x)$ for the localized solution as follows
\begin{equation}\label{R2}
R(x) = \pm \left(\frac{m\cos\left(2\theta(x)\right)-\omega}
  {s \cos^{k+1}\left(2\theta(x)\right) + v }\right)^{{1}/{2k}}.
\end{equation}

Hereto, the standing wave solution of the
NLD equation \eqref{generalnld} with \eqref{generalLI} has been derived,
and will be denoted as follows
\begin{equation}
\vPsi^\text{sw}(x,t) = \left(
    \begin{array}{c}
      \psi^\text{sw}_1(x,t) \\ \psi^\text{sw}_2(x,t)
    \end{array}\right)
    = \me^{-\mi \omega t} R(x)\left(
      \begin{array}{c}
        \cos\big(\theta(x)\big) \\ \mi \sin\big(\theta(x)\big)
      \end{array} \right),   \label{eq:bound-state}
\end{equation}
where $\theta(x)$ and $R(x)$ are given in Eqs. \eqref{theta-1} and \eqref{R2}, respectively.
This solution represents a solitary wave with zero velocity and
contains some special cases in the literature \eg
\cite{AlvarezCarreras1981,CooperKhareMihailaSaxena2010}.

\begin{Remark}\rm
It has been pointed out in \cite{ShaoTang2005}
that the profile of the charge density $J_0(x,t)$
for the standing wave \eqref{eq:bound-state} under the scalar self-interaction (\ie $s\neq 0$ and $v=0$)
with $k=1$
can be either one-humped or two-humped,
which is also recently confirmed for any $k>0$ by other researchers
in \cite{CooperKhareMihailaSaxena2010}.
They further pointed out there that
the profile can only be one-humped for any $k>0$ in the case of $s=0$ and $v\neq 0$.
For the linear combined self-interaction \eqref{generalLI} with $s\neq 0$ and $v\neq0$,
we find that the profile can also be one-humped (see Figs.~\ref{fig.os-dg-rhoq} and \ref{fig:k1-sv}
where the charge density is denoted by $\rho_Q(x,t)$ for convenience)
or two-humped (see Fig.~\ref{fig:k1-2humped-sv}).
\end{Remark}

\subsection{Solitary wave solution with nonzero velocity}
\label{boostforMW}

This subsection will derive another solution of the (1+1)-dimensional NLD equation \eqref{generalnld} with the self-interaction \eqref{generalLI} by using the Lorentz covariance of the NLD equation. Consider a frame F with an observer O
and coordinates $(x,t)$. The observer O
describes a particle by the wavefunction $\vPsi(x,t)$
which obeys the NLD equation \eqref{generalnld} with $L_\text{I}$ given in \eqref{generalLI}, \ie
\begin{equation}\label{nld-lorentz}
\left(\mi \vgamma^0\pd{}{t}+\mi \vgamma^1\pd{}{x}-m\right)\vPsi(x,t)
+\left(\partial L_\text{I}[\vPsi]/\partial \wbar{\vPsi}\right)(x,t)=0.
\end{equation}
In another inertial frame F$'$ with an observer O$'$ and coordinates $(x',t')$ given by
 the Lorentz transformation  with a translation in the $x$-direction
\begin{align}\label{boost}
x'=\gamma \big((x-x_0)-Vt\big), \quad t'=\gamma\big(t-V(x-x_0)\big),
\end{align}
which is called ``boost'' in the $x$--direction,
where $x_0$ is  any given  position,
$V$ is the relative velocity between frames in the $x$-direction,
and $\gamma=1/\sqrt{1-V^2} $ is the Lorentz factor. 
According to the relativity principle,
the observer O$'$ describes the same particle by  $\vPsi'(x',t')$ which should also satisfy
\begin{equation}\label{nld-prime}
\left(\mi \vgamma^0\pd{}{t'}+\mi \vgamma^1\pd{}{x'}-m\right)\vPsi'(x',t')
+\left(\partial L_\text{I}'[\vPsi']/\partial \wbar{\vPsi}'\right)(x',t')=0.
\end{equation}
Using some algebraic manipulations, the ``transformation'' matrix $\vec{S}$
may be found as follows
\begin{align}\label{boost2b}
    \vec{S}
    = \left(
    \begin{matrix}
      \sqrt{\frac{\gamma+1}{2}} & \text{sign}(V)\sqrt{\frac{\gamma-1}{2}} \\
      \text{sign}(V)\sqrt{\frac{\gamma-1}{2}} & \sqrt{\frac{\gamma+1}{2}}
    \end{matrix}
    \right),
\end{align}
which takes $\vPsi(x,t)$ to $\vPsi'(x',t')$ under the Lorentz transformation
\eqref{boost}, \ie
\begin{align} \label{boost2}
\vPsi'(x',t^\prime) = \vec{S} \vPsi(x,t),
\end{align}
where   $\text{sign} (x)$ is the sign function
which returns $1$ if $x>0$, $0$ if $x=0$,
and $-1$ if $x<0$.

Applying the transformation \eqref{boost2} with \eqref{boost2b} to
 the standing wave solution \eqref{eq:bound-state}  gives
another solution of
the NLD equation \eqref{generalnld}-\eqref{generalLI},
\ie the moving wave solution
\begin{align}
\vPsi^\text{mw}(x-x_0,t) = &
\left(
\begin{matrix}
\sqrt{\frac{\gamma+1}{2}} \psi^\text{sw}_1(x',t')  + \mathrm{sign}(V)\sqrt{\frac{\gamma-1}{2}} \psi^\text{sw}_2(x',t')\\
\mathrm{sign}(V)\sqrt{\frac{\gamma-1}{2}} \psi^\text{sw}_1(x',t') + \sqrt{\frac{\gamma+1}{2}} \psi^\text{sw}_2(x',t')
\end{matrix}
\right). \label{mw}
\end{align}
This solution represents a NLD solitary wave with velocity $V$
and will return to the standing wave \eqref{eq:bound-state} if setting $V=0$ and $x_0=0$.

\subsection{Time reversibility}

This subsection will show that the NLD equation \eqref{generalnld} with $L_\text{I}$ given in \eqref{generalLI}
remains invariant under the time reversal operation
\begin{equation} \label{TimeR}
x^\prime = x, \quad t^\prime = -t, \quad
\vPsi^\prime(x^\prime,t^\prime) = \vK \vPsi(x,t), \quad \vK=\vgamma^0 \mathcal{C}, 
\end{equation}
where $\mathcal{C}$ denotes the complex conjugate operation on $\vPsi(x,t)$, \ie $\vPsi^\ast(x,t)= \mathcal{C}\vPsi(x,t)$,
the time-reversal operator $\vK$ satisfies
\begin{equation}\label{k-relation}
\vK^\dag\vK=\vI, \quad \vK\vgamma^0 = \vgamma^0 \vK, \quad  \vK\vgamma^1 = - \vgamma^1 \vK,
\end{equation}
due to the anticommutation relation $\{\vgamma^\mu,\vgamma^\nu\}=2 \eta^{\mu\nu}\vI$,
and $\vI$ is the unit matrix.
From the relations \eqref{k-relation},
it can be easily verified that
\begin{equation}\label{tilde-relation}
\begin{aligned}
 (\wbar{\vPsi}^\prime \vPsi^\prime)(x^\prime,t^\prime)  &= (\wbar{{\vPsi}} {\vPsi})(x,t), \\
 (\wbar{\vPsi}^\prime \vgamma^0 \vPsi^\prime)(x^\prime,t^\prime)  &= (\wbar{{\vPsi}} \vgamma^0 {\vPsi})(x,t), \\
 (\wbar{\vPsi}^\prime \vgamma^1 \vPsi^\prime)(x^\prime,t^\prime)  &= - (\wbar{{\vPsi}} \vgamma^1 {\vPsi})(x,t),
\end{aligned}
\end{equation}
so that the self-interaction Lagrangian in \eqref{generalLI} satisfies
\begin{equation}\label{tr-L-relation}
L_\text{I}^\prime[\vPsi^\prime](x^\prime,t^\prime) = L_\text{I}[\vPsi](x,t),
\end{equation}
under the time reversal transformation \eqref{TimeR}.

Applying the time-reversal operator $\vK$ to the NLD equation \eqref{nld-lorentz} and using
the definition \eqref{TimeR} as well as the relations \eqref{k-relation} and \eqref{tr-L-relation}
lead to an equation which has the same form as shown in \eqref{nld-prime}.
That is, if a spinor $\vPsi(x,t)$ satisfies the NLD equation \eqref{nld-lorentz},
then the transformed spinor $\vPsi^\prime(x^\prime,t^\prime)$
by the time reversal operation \eqref{TimeR} will also
satisfy the same NLD equation,
\ie the NLD equation \eqref{nld-lorentz} is time reversible under the operation \eqref{TimeR}.

\section{Numerical methods}
\label{sec:DiffSch}

As we mentioned in Section \ref{sec:intro},
some numerical methods have been well developed for
the NLD equation with a scalar or vector self-interaction.
This section will  extendedly present and compare several numerical methods
for solving
the (1+1)-dimensional NLD equation \eqref{eq:gnld}
with the general scalar and vector self-interaction \eqref{generalLI}.
Their numerical analyses
will be presented in Section \ref{sec:discussion4fdm}.

For the sake of convenience, divide the bounded spatial domain $\Omega\subset \mathbb R$ into
a uniform partition $\{x_j | x_j=jh \in \Omega, j\in \mathbb Z\}$ with a constant stepsize $h>0$
and give a grid in time
$\{t_n = n\tau, n=0,1,\cdots\}$ with a time stepsize $\tau>0$,
and recast the NLD equation \eqref{nld-lorentz} into
\begin{equation}
  \pp{\vPsi}{t} +\vgamma^0\vgamma^1 \pp{\vPsi}{ x}
  + \mi m \vgamma^0 \vPsi -  \mi f_{\text{s}} \vgamma^0\vPsi
- \mi f_{\text{v}}
 \wbar{\vPsi}\vgamma_\mu\vPsi \vgamma^0\vgamma^\mu\vPsi
= 0,  \label{eq:gnld}
\end{equation}
where
\begin{equation*}
f_{\text{s}}  := s (k+1)w_{\text{s}}^{k}, \quad
w_{\text{s}} := \wbar{\vPsi}\vPsi, \quad
f_{\text{v}}  := v (k+1) w_{\text{v}}^{(k-1)/2},\quad
w_{\text{v}} := \wbar{\vPsi}\vgamma^\nu\vPsi \wbar{\vPsi}\vgamma_\nu\vPsi, \quad
\end{equation*}
all of which are real functions, and the dependence of $\vPsi(x,t)$ on $(x,t)$ is implied.

\subsection{Several finite difference schemes}

Use $\vec \Psi_j^n$ to denote  approximation of $\vec \Psi(x_j,t_n)$ and
define the forward,  backward and centered difference operators in space and time  by:
\begin{equation}\label{eq:note-sch}
\begin{aligned}
\vdelta_t^\pm =\pm (\vE_t^{\pm 1} -I)/\tau,  \quad \vdelta_t^0 =  (\vE_t -\vE_t^{-1})/2\tau,
\\
\vdelta_x^\pm =\pm (\vE_x^{\pm 1} -I)/h,\quad \vdelta_x^0 =  (\vE_x -\vE_x^{-1})/2h,
\end{aligned}
\end{equation}
where $I$ is the identity operator,
and $\vE_t$ and $\vE_x$ are the translation operators in time and space, respectively, defined by
$$
\vE_t \vPsi^n_j:=\vPsi^{n+1}_j,\quad \vE_x \vPsi^n_j:=\vPsi^n_{j+1},
$$
whose inverses exist and are defined by
$$
\vE_t^{-1} \vPsi^n_j:=\vPsi^{n-1}_j,\quad \vE_x^{-1} \vPsi_j^n:=\vPsi_{j-1}^n.
$$
Besides,  several symbols are also introduced for arithmetic averaging operators:
\begin{equation*}
\begin{aligned}
      \vell_t^\pm\vPsi_j^n &:=  \frac12(\vPsi_j^{n\pm1} + \vPsi_j^n),\quad
      \vell_t^0 \vPsi_j^n :=  \frac12(\vPsi_j^{n+1} +\vPsi_j^{n-1}),\\
      \vell_x^\pm \vPsi_j^{n} &:= \frac1{2}(\vPsi_{j\pm 1}^n +\vPsi_{j}^n),\quad
      \vell_x^0 \vPsi_j^{n} := \frac1{2}(\vPsi_{j+1}^n + \vPsi_{j-1}^n),
\end{aligned}
\end{equation*}
and for an extrapolation operator:
\begin{equation*}
        \vell_t^\text{e}\vPsi_j^{n} = \frac32\vPsi_j^n - \frac12\vPsi_j^{n-1}.
\label{eq:note-sch-sum}
\end{equation*}

\textbf{$\bullet$ Crank-Nicolson schemes}
The CN scheme and its linearized version have been studied in \cite{AlvarezKuoVazquez1983,Alvarez1992,Jimcnez1994}
for the NLD equation with the quadric scalar self-interaction.
For the system \eqref{eq:gnld}, the extension of the CN scheme in \cite{Jimcnez1994,ChangJiaSun1999} may be written as
\begin{align}
\begin{aligned}
\vdelta_t^+\vPsi_j^n +  \vgamma^0\vgamma^1 \vell_t^+\vdelta_x^0 \vPsi_j^n
    &+ \mi m \vgamma^0 \vell_t^+\vPsi_j^{n} -
    \mi \frac{\vdelta_t^+(F_{\text{s}})_j^n}
    {\vdelta_t^+(w_{\text{s}})_j^n}\vgamma^0\vell_t^+\vPsi_j^n \\
    & -  \mi \frac{\vdelta_t^+(F_{\text{v}})_j^n}
    {\vdelta_t^+(w_{\text{v}})_j^n}
    \vell_t^+(\wbar{\vPsi}\vgamma_\mu\vPsi)_j^n\vgamma^0\vgamma^\mu \vell_t^+\vPsi_j^n
    = 0,
\end{aligned} \label{eq:cn-gnld}
\end{align}
by approximating \eqref{eq:gnld} at point $(x_{j},t_{n+\frac12})$ with compact
central difference  quotient in place of the partial derivative,
where
$$F_I(w_I) := \int_0^{w_I} f_I(x) \dif x, \quad  I =\text{s}, \text{v}.$$
The above CN scheme (named as {\tt CN} hereafter) is fully implicit and forms a nonlinear algebraic system.
In practice, the linearization technique is often used to
overcome difficulty in directly solving such nonlinear algebraic system.
Two linearization techniques \cite{ChangJiaSun1999} for numerical methods of the nonlinear
Schr\"{o}dinger equation are borrowed here.
The first linearized CN scheme we consider is using wholly the extrapolation
technique to the nonlinear self-interaction terms in \eqref{eq:gnld}
\begin{equation}  \label{eq:lcn1-gnld}
     \delta_t^+\vPsi_j^{n} + \vgamma^0\vgamma^1\ell_t^+\delta_x^0\vPsi_j^n
    + \mi m \vgamma^0 \ell_t^+\vPsi_j^{n}
    - \mi \ell_t^\text{e}\left(   f_{\text{s}}\vgamma^0\vPsi
      +   f_{\text{v}} \wbar{\vPsi}\vgamma_\mu\vPsi\vgamma^0\vgamma^\mu\vPsi \right)_j^n = 0,
\end{equation}
which   will be called  by {\tt LCN1}.
The second linearized CN scheme, denoted  by {\tt LCN2}, is
\begin{equation}
\begin{aligned}  \label{eq:lcn2-gnld}
      \delta_t^+\vPsi_j^{n} +  \vgamma^0\vgamma^1\ell_t^+\delta_x^0\vPsi_j^n
    & + \mi m \vgamma^0 \ell_t^+\vPsi_j^{n}
    -   \mi \ell_t^\text{e}(f_{\text{s}})_j^n\vgamma^0\ell_t^+\vPsi_j^n \\
       & -  \mi \ell_t^\text{e}(f_{\text{v}}\wbar{\vPsi}\vgamma_\mu\vPsi)_j^n\vgamma^0\vgamma^\mu \ell_t^+\vPsi_j^n = 0,
\end{aligned}
\end{equation}
by  partially applying the extrapolation technique to the nonlinear self-interaction terms. 
It is worth noting that the above linearized CN schemes are not linearized version of the {\tt CN} scheme \eqref{eq:cn-gnld}.
The {\tt LCN2} scheme may conserve the charge  and behaves better than the {\tt LCN1} scheme (\videpost).

\begin{Remark}\rm
For the (1+1)-dimensional NLD equation \eqref{generalnld}
with the quadric scalar self-interaction Lagrangian \eqref{self-scalar1}, 
the CN scheme (named by {\tt CN0})
proposed in \cite{AlvarezKuoVazquez1983} is
\begin{equation}
  \label{eq:cn-kuo}
      \delta_t^+\vPsi_j^n + \vgamma^0\vgamma^1\ell_t^+\delta_x^0\vPsi_j^n
    + \mi m \vgamma^0 \ell_t^+\vPsi_j^{n} - 2 {\mi}  s
    (\ell_t^+\wbar{\vPsi}_j^n\ell_t^+\vPsi_j^n)
    \vgamma^0\ell_t^+\vPsi_j^n = 0,
\end{equation}
 and its linearized version (called by {\tt LCN0}) is given in \cite{Alvarez1992} as follows
\begin{equation}
\begin{aligned}  \label{eq:lcn-kuo}
        \delta_t^+\vPsi_j^n & + \vgamma^0\vgamma^1\ell_t^+\delta_x^0\vPsi_j^n
    + \mi m \vgamma^0 \ell_t^+\vPsi_j^{n} \\
    & - 2 {\mi} s \left((\wbar{\vPsi}\vPsi)_j^n
     - \tau\real( \wbar{\vPsi}_j^n\vgamma^0\vgamma^1\delta_x^0\vPsi_j^n)\right)
    \vgamma^0\ell_t^+\vPsi_j^n = 0.
\end{aligned}
\end{equation}
We will show in Section \ref{sec:discussion4fdm} that the {\tt CN},  {\tt CN0}
and {\tt LCN0} schemes conserve the charge and the {\tt CN} scheme further conserves the energy.
\end{Remark}

\textbf{$\bullet$ Odd-even hopscotch scheme} The odd-even hopscotch scheme is a numerical integration technique
for time-dependent partial differential equations,
see \cite{Gordon1965}. Its key point is  that the forward Euler-central difference scheme
is used for the odd grid points while at the even points the backward Euler-central
difference scheme is recovered. Thus the odd-even hopscotch scheme may be explicitly
implemented.
Such scheme applied to the system \eqref{eq:gnld} becomes
\begin{equation}
\begin{aligned}  \label{eq:hops-odd}
     \delta_t^+\vPsi_j^{n} & + \vgamma^0\vgamma^1\delta_x^0\vPsi_j^{n}
    + \mi m \vgamma^0 \vPsi_j^{n} \\
    & -  \mi \ell_x^0\left(f_{\text{s}}\vgamma^0\vPsi
     +   f_{\text{v}} \wbar{\vPsi}\vgamma_\mu\vPsi\vgamma^0\vgamma^\mu\vPsi \right)_j^n = 0, \quad  \text{$n+j$ is odd,}
\end{aligned}
\end{equation}
\begin{equation}
\begin{aligned}  \label{eq:hops-even}
     \delta_t^-\vPsi_j^{n+1} & + \vgamma^0\vgamma^1\delta_x^0\vPsi_j^{n+1}
    + \mi m \vgamma^0 \vPsi_j^{n+1} \\
    & - \mi \ell_x^0\left(f_{\text{s}}\vgamma^0\vPsi
    +   f_{\text{v}} \wbar{\vPsi}\vgamma_\mu\vPsi\vgamma^0\vgamma^\mu\vPsi \right)_j^{n+1} = 0,
    \quad \text{$n+j$ is even.}
\end{aligned}
\end{equation}
In the following we will call it by {\tt HS}.

\textbf{$\bullet$ Leapfrog scheme} The leapfrog scheme looks quite similar to the forward scheme, see \eg \eqref{eq:hops-odd},
 except it uses the values  from the previous time-step instead of the current one.
For the system \eqref{eq:gnld}, it is
\begin{equation}
  \label{eq:3l-ex-1}
    \delta_t^0\vPsi_j^{n} +\vgamma^0\vgamma^1\delta_x^0\vPsi_j^{n}
    + \mi m\vgamma^0 \vPsi_j^{n}
    -  \mi \left(f_{\text{s}}\vgamma^0\vPsi
    +   f_{\text{v}} \wbar{\vPsi}\vgamma_\mu\vPsi\vgamma^0\vgamma^\mu \vPsi\right)_j^n = 0,
\end{equation}
which is a three-level explicit scheme in time with a central difference in space
and  will be named  by {\tt LF}.

\textbf{$\bullet$ Semi-implicit scheme} Another three-level scheme considered here for the system \eqref{eq:gnld} is
\begin{equation}
  \label{eq:semi-ex}
    \delta_t^0\vPsi_j^{n} + \vgamma^0\vgamma^1\ell_t^0\delta_x^0\vPsi_j^{n}
    + \mi m \vgamma^0 \ell_t^0\vPsi_j^{n}
    -  \mi \left(f_{\text{s}}\vgamma^0\vPsi
    +  f_{\text{v}}
    \wbar{\vPsi}\vgamma_\mu\vPsi\vgamma^0\vgamma^\mu \vPsi\right)_j^n = 0,
\end{equation}
which
is obtained by approximating explicitly the nonlinear terms but implicitly the linear terms and  will be called  by {\tt SI}.
It is worth noting that such semi-implicit scheme has been studied for the NLD equation with the quadric scalar self-interaction in \cite{Bournaveas2012}.

\subsection{Exponential operator splitting scheme}
\label{sec:split-tech}

This subsection goes into discussing exponential operator splitting scheme
for the NLD equation \eqref{eq:gnld}.
For convenience, we rewrite the system \eqref{eq:gnld} as follows
\begin{equation}
  \vPsi_t = \left(\mathcal{L} + \mathcal{N}_{\text{s}}
+ \mathcal{N}_{\text{v}}\right)\vPsi,
\label{eq:split-4}
\end{equation}
where {the linear operator $\mathcal{L}$ and  both nonlinear operators $\mathcal{N}_{\text{s}}$
and $\mathcal{N}_{\text{v}}$  are defined by}
\begin{equation*}
\mathcal{L} \vPsi  := -\vgamma^0\vgamma^1 \vPsi_x -\mi m\vgamma^0\vPsi,\quad
\mathcal{N}_{\text{s}} \vPsi :=   \mi f_{\text{s}}\vgamma^0\vec \Psi,\quad
\mathcal{N}_{\text{v}} \vPsi :=   \mi f_{\text{v}} \wbar{\vPsi}\vgamma_\mu\vPsi \vgamma^0\vgamma^\mu\vPsi.
\end{equation*}
Then the problem \eqref{eq:split-4} may be decomposed into three subproblems as follows
\begin{align}
  \vPsi_t &= \mathcal{L} \vPsi,
  \label{eq:split-l-2}
  \\
  \label{eq:split-ns-2}
  \vPsi_t &= \mathcal{N}_{\text{s}} \vPsi,
  \\ \label{eq:split-nv-2}
  \vPsi_t &= \mathcal{N}_{\text{v}} \vPsi.
\end{align}
{Due to the local conservation laws which are discussed in
Section \ref{sec:dis-nonlinear}, solutions of the nonlinear subproblem \eqref{eq:split-ns-2} or \eqref{eq:split-nv-2}
 may be expressed as an exponential of the operator $\mathcal{N}_{\text{s}}$ or $\mathcal{N}_{\text{v}}$
 acting on ``initial data''. Thus we may introduce the exponential operator splitting scheme for the the NLD equation \eqref{eq:split-4} or \eqref{eq:gnld}, imitating that for the linear partial differential equations, see \eg \cite{Sheng1989,Chin2005} and references therein.} Based on the exact or approximate solvers of those three subproblems, a more
general $K$-stage $N$-th order exponential operator splitting method \cite{Sornborger1999,ThalhammerCaliariNeuhauser2009}
for the system \eqref{eq:split-4} can be cast into {product of  finitely many exponentials as follows}
\begin{equation}
\label{eq:high-oder-split}
  \vPsi_j^{n+1} = \prod_{i=1}^{K}
  \big(
  \exp({\tau_{i} \mathcal{A}_i^{(1)}})\exp({\tau_{i}\mathcal{A}_i^{(2)}})
  \exp({\tau_{i} \mathcal{A}_i^{(3)}})\big)  \vPsi_j^{n},
\end{equation}
where $\tau_i$ denotes the time stepsize used within the $i$-th stage
and satisfies
$
\sum_{i=1}^K \tau_{i}=\tau,
$
and   $\{\mathcal{A}_i^{(1)},\mathcal{A}_i^{(2)},\mathcal{A}_i^{(3)}\}$ is
any permutation of
$\{\mathcal{L},\mathcal{N}_{\text{s}},\mathcal{N}_{\text{v}}\}$.
Hereafter we call the operator splitting scheme \eqref{eq:high-oder-split} by {\tt OS($N$)}.
{Although one single product of finitely many exponentials exponentials \eqref{eq:high-oder-split}
is employed here, it should be pointed out that the
linear combination of such finite products can also be used to construct
exponential operator splitting schemes as shown in \cite{Sheng1989}.}

A simple example is the well-known second-order accurate operator splitting method
of Strang (named by {\tt OS($2$)}) with
\begin{equation}
\label{eq:strang-split}
K=2, \quad \tau_1=\tau_2{=\frac12\tau}, \quad \mathcal{A}_1^{(1)}=\mathcal{A}_2^{(3)},\quad
\mathcal{A}_1^{(3)}=\mathcal{A}_2^{(1)}, \quad
\mathcal{A}_1^{(2)}=\mathcal{A}_2^{(2)}.
\end{equation}
Another example is the fourth-order accurate  operator splitting method
\cite{Sornborger1999}
with
\begin{equation*}
\begin{aligned}
K=&8,\quad
\mathcal{A}_q^{(1)}=\mathcal{A}_p^{(3)},\quad
\mathcal{A}_q^{(3)}=\mathcal{A}_p^{(1)}, \quad
\mathcal{A}_q^{(2)}=\mathcal{A}_p^{(2)}, \quad q=1,4,6,7, p=2,3,5,8,\\
\tau_1=&\tau_8=\frac{\tau}{5-\sqrt{13}+\sqrt{2(1+\sqrt{13})}},
\quad
\tau_2=\tau_7=\frac{7+\sqrt{13}-\sqrt{2(1+\sqrt{13})}}{24}\tau,\\
\tau_3=&\tau_6=\frac{\tau_1^2}{\tau_2-\tau_1},\quad
\tau_4=\tau_5=\frac{\tau_2(\tau_1-\tau_2)}{3\tau_1-2\tau_2},
\end{aligned}
\end{equation*}
which is denoted by {\tt OS($4$)} in the following.

\begin{Remark}\rm
Another operator splitting scheme is studied
in \cite{FrutosSanz-Serna1989}
for the NLD equation (\ref{generalnld})
but only with the quadric scalar self-interaction Lagrangian,
and the second-order accurate Strang method \eqref{eq:strang-split} is applied there.
For the system \eqref{eq:gnld}, it is based on
the following operator decomposition
\begin{equation}
  \vPsi_t = \left(\widehat{\mathcal{L}} + \widehat{\mathcal{N}}_{\text{s}}
+ \widehat{\mathcal{N}}_{\text{v}}\right)\vPsi,
 \label{eq:split-l}
\end{equation}
with
\begin{equation*}    
\widehat{\mathcal{L}} \vPsi  := -\vgamma^0\vgamma^1 \vPsi_x,\quad
\widehat{\mathcal{N}}_{\text{s}} \vPsi :=  -  \mi \left(m - f_{\text{s}}\right)\vgamma^0\vPsi,\quad
\widehat{\mathcal{N}}_{\text{v}} \vPsi :=   \mi f_{\text{v}} \wbar{\vPsi}\vgamma_\mu\vPsi \vgamma^0\vgamma^\mu\vPsi.
\end{equation*}
\end{Remark}

\begin{Remark}\rm
{For the linear parabolic equation which is an irreversible system,
 a more general exponential operator splitting scheme  and its accuracy as well as stability are discussed  in \cite{Sheng1989}, based on  linear combinations of products of  finitely many exponentials.
It is shown that for such irreversible system,  negative weights or negative time stepsizes $\tau_i$  may lead to instability;
and the highest order of the stable exponential operator splitting approximation (only with positive weights and
positive sub-stepsizes in time) is two.
However, for time-reversible systems, such as the Hamilton system, the Schr{\"o}dinger equations and
the NLD equation \eqref{generalnld} with $L_\text{I}$ given in \eqref{generalLI}, it
is immaterial whether or not the weights or time stepsizes $\tau_i$  are positive \cite{Chin2005}, where
 a general framework was presented for understanding the structure of the exponential operator splitting schemes
and both specific error terms and order conditions were analytically solved.
}
\end{Remark}

\subsubsection{Linear subproblem}
\label{sec:dis-linear}

We are now solving the the linear subproblem   \eqref{eq:split-l-2}.
Denote its ``initial data'' by $\vPsi_j^{(0)}=\big((\psi_1)_j^{(0)},(\psi_2)_j^{(0)}\big)^T$  at the $i$-th stage in \eqref{eq:high-oder-split}.

If the spinor $\vPsi$ is periodic (\eg $2\pi$-periodic)  with respect to $x$, the Fourier spectral method is employed
to solve \eqref{eq:split-l-2} and
gives
\begin{equation}  \label{eq:fft}
    \vPsi_j^{(1)} = \mathcal{F}^{-1}\Big(
      \exp\big({-\mi \tau_i(\kappa \vec\gamma^0\vec\gamma^1+m\vec\gamma^0) }\big)
                      \mathcal{F}\left(\vPsi_j^{(0)}\right)\Big).
\end{equation}
Here $\mathcal{F}$ and $\mathcal{F}^{-1}$ denote the discrete Fourier transform operator and its inverse,
respectively,  defined by
\begin{align*}
   \label{eq:fft}
\big(\mathcal{F}(\vPsi)\big)_\kappa :=& \sum_{j=0}^{J-1}\vPsi_j\exp\big({-\mi 2\pi\kappa \frac{j}J}\big)
\quad \kappa = 0,\ldots, J-1, \\
\big(\mathcal{F}^{-1}(\vPhi)\big)_j :=& \frac1{J}\sum_{\kappa=0}^{J-1}\vPhi_\kappa\exp\big({\mi 2\pi j \frac{\kappa}J}\big)
\quad j = 0,\ldots, J-1,
\end{align*}
where  $J$ is the grid point number,
and the matrix exponential in \eqref{eq:fft} can be easily evaluated as follows
\begin{equation}\label{eq:matrixfft}
 \exp\big({-\mi \tau_i(\kappa \vec\gamma^0\vec\gamma^1 +m\vec\gamma^0) } \big)= \begin{pmatrix}
     \cos (\zeta \tau_{i}) -\mi\frac{m}{\zeta}\sin (\zeta \tau_{i})&
     -\mi\frac{\kappa }{\zeta}\sin (\zeta \tau_{i}) \\
     -\mi\frac{\kappa }{\zeta}\sin (\zeta \tau_{i}) &
     \cos (\zeta \tau_{i}) +\mi\frac{m}{\zeta}\sin(\zeta \tau_{i})
   \end{pmatrix},
\end{equation}
with $\zeta = \sqrt{\kappa^2+m^2}$.

When the spinor $\vPsi$ is not periodic  with respect to $x$, the fifth-order accurate finite difference WENO scheme  will be used to
solve the linear subproblem  \eqref{eq:split-l-2}. The readers are referred to \cite{Shu2009-WENO} for  details.
In this case, the linear subproblem  \eqref{eq:split-l-2} can also be solved by using
the characteristics method.

\subsubsection{Nonlinear subproblems}
\label{sec:dis-nonlinear}

The nonlinear subproblems \eqref{eq:split-ns-2} and  \eqref{eq:split-nv-2}
are left to be solved now.
Their ``initial data'' is still denoted by $\vec\Psi_j^{(0)}=\big((\psi_1)_j^{(0)},(\psi_2)_j^{(0)}\big)^T$
   at the $i$-th stage in \eqref{eq:high-oder-split}, and define
$$
t_n^{(i)}=t_n+\sum_{p=1}^{i-1} \tau_p,\quad i=1,2,\cdots,K.
$$

For nonlinear subproblem  \eqref{eq:split-ns-2},   it is not difficulty to verify that
$\partial_t w_\text{s} = 0$ so that
\begin{equation}\label{ns-conv}
\partial_t f_\text{s} = 0.
\end{equation}
Using this local conservation law gives the solution at $t=t_n^{(i+1)}$ of \eqref{eq:split-ns-2}
with the ``initial data'' $\vec\Psi_j^{(0)}$ as follows
\begin{align}
\vPsi_j^{(1)}
&= \exp\left(\mi \int_{t_n^{(i)}}^{t_n^{(i+1)}} (f_{\text{s}})_j\vgamma^0 \dif t \right) \vPsi_j^{(0)}  = \exp\left(\mi (f_{\text{s}})_j^{(0)}\vgamma^0 \tau_{i} \right) \vPsi_j^{(0)} \nonumber \\
&= \diag\left\{\exp\left( \mi (f_{\text{s}})_j^{(0)} \tau_{i}\right),
\exp\left(-\mi(f_{\text{s}})_j^{(0)} \tau_{i}\right)\right\} \vPsi_j^{(0)}. \label{OS-NS-eq3}
\end{align}

For the nonlinear subproblem \eqref{eq:split-nv-2}, one may still similarly
derive the following local conservation laws
\begin{equation}\label{nv-con}
\partial_t (\wbar{\vPsi}\vgamma_0\vPsi) = 0,\quad
\partial_t (\wbar{\vPsi}\vgamma_1\vPsi) = 0, \quad
\partial_t f_\text{v} = 0.
\end{equation}
by direct algebraic manipulations if using the fact that
$\wbar{\vec \Psi}\vec\gamma_0\vec \Psi$,
$\wbar{\vec \Psi}\vec\gamma_1\vec \Psi$
and $f_{\text{v}}$ are all real.
Consequently, integrating \eqref{eq:split-nv-2} from $t_n^{(i)}$ to $t_n^{(i+1)}$ gives
its solution  as follows
\begin{align}
  \vPsi_j^{(1)} & = \exp\left(\mi(f_{\text{v}} \wbar{\vPsi}\vgamma_\mu\vPsi)_j^{(0)}\vgamma^0\vgamma^\mu
   \tau_{i} \right)\vPsi_j^{(0)} \nonumber \\
   & =\exp(\mi{\alpha\tau_i}) \left(
  \begin{array}{rr}
    \cos(\beta\tau_i) & \mi \sin (\beta\tau_i) \\
    \mi \sin (\beta\tau_i) &\cos (\beta\tau_i)
  \end{array}\right) \vPsi_j^{{(0)}},  \label{eq:solution-V}
\end{align}
where $\alpha = (f_{\text{v}}  \wbar{\vPsi}\vgamma_0\vPsi)_j^{(0)}$ and $\beta = (f_{\text{v}}  \wbar{\vPsi}\vgamma_1\vPsi)_j^{(0)}$.

\begin{Remark}\rm
It is because the local conservation laws \eqref{ns-conv} and \eqref{nv-con} are fully exploited here that
we can solve exactly the nonlinear subproblems \eqref{eq:split-ns-2} and \eqref{eq:split-nv-2} which imply
the more higher accuracy of the OS method than that of other methods.
\end{Remark}

In summary, we have
\begin{itemize}
\item The {\tt CN} \eqref{eq:cn-gnld} and {\tt CN0} \eqref{eq:cn-kuo} schemes are nonlinear and implicit,
and could be solved by iterative algorithms such as Picard iteration and Newton method.
\item The {\tt LCN0}~\eqref{eq:lcn-kuo}, {\tt LCN1}~\eqref{eq:lcn1-gnld}, {\tt LCN2}~\eqref{eq:lcn2-gnld} and {\tt SI}~\eqref{eq:semi-ex}
schemes are linear and implicit.  
\item The {\tt HS} \eqref{eq:hops-odd}-\eqref{eq:hops-even}, {\tt LF} \eqref{eq:3l-ex-1}, and {\tt OS($N$)} \eqref{eq:high-oder-split} schemes are explicit.
\end{itemize}

\section{Numerical analysis}
\label{sec:discussion4fdm}

Before investigating the performance of the numerical methods proposed in Section \ref{sec:DiffSch},
this section will go first into numerical analysis of them,
including the accuracy in the sense of the truncation error, time reversibility and the conservation of the charge or energy.

\begin{Proposition}
If $\vPsi(x,t)\in C^{\infty}(\mathbb{R}\times[0,+\infty))$ is periodic, 
then the {\tt CN}, {\tt CN0}, {\tt LCN0}, {\tt LCN1}, {\tt LCN2}, {\tt HS}, {\tt LF} and {\tt SI} schemes
  are  of order $\mathcal{O}(\tau ^2 + h^2)$,
and the {\tt OS($N$)} scheme is of order $\mathcal{O}(\tau^N + h^m)$ for any arbitrary large $m>0$.
\end{Proposition}
\begin{proof}
The proof is very straightforward by using directly the Taylor series expansion for the finite difference schemes
and the Fourier spectral analysis for the {\tt OS($N$)} scheme, and thus is skipped here for saving space.
\end{proof}

\begin{Proposition}
  The {\tt CN}, {\tt CN0}, {\tt HS}, {\tt LF}, {\tt SI}, and {\tt OS($N$)} schemes are time reversible,
  but the {\tt LCN0}, {\tt LCN1}, {\tt LCN2} schemes are not.
\end{Proposition}
\begin{proof}
We give the proof for the {\tt CN} and {\tt LCN1} schemes as an example and the others can be proved in a similar way.

According to the transformation \eqref{TimeR},
the relation between the transformed finite difference solution and the original
one should be ${(\vPsi^\prime)}_j^{n^\prime}=\vK\vPsi_j^n = (\vK\vPsi)_j^n$ with $n^\prime=-n$.
Consequently,
we have
\begin{equation}\label{TR-e1}
\begin{aligned}
\vK \delta_t^+\vPsi_j^n &= \vK \frac{\vPsi_j^{n+1}-\vPsi_j^n}{\tau} = \frac{{(\vPsi^\prime)}_j^{n^\prime-1}-{(\vPsi^\prime)}_j^{n^\prime}}{\tau} = - \delta_t^+ {(\vPsi^\prime)}_j^{n^\prime-1},
\\
\vK \ell_t^+\vPsi_j^n  &= \vK \frac{\vPsi_j^{n+1}+\vPsi_j^n}{2} = \frac{{(\vPsi^\prime)}_j^{n^\prime-1}
+{(\vPsi^\prime)}_j^{n^\prime}}{2} = \ell_t^+ {(\vPsi^\prime)}_j^{n^\prime-1}.
\end{aligned}
\end{equation}
and then using the relations in \eqref{tilde-relation} yields
\begin{equation}\label{TR-e2}
\begin{aligned}
\ell_t^+(\wbar{\vPsi}\vgamma_0\vPsi)_j^n &= \ell_t^+(\wbar{\vPsi}^\prime\vgamma_0{\vPsi}^\prime)_j^{n^\prime-1}, \\
\ell_t^+(\wbar{\vPsi}\vgamma_1\vPsi)_j^n &= - \ell_t^+(\wbar{\vPsi}^\prime\vgamma_1{\vPsi^\prime})_j^{{n}^\prime-1},\\
\delta_t^+(w_I)_j^n &= -\delta_t^+(w_I^\prime)_j^{n^\prime-1}, \\
\delta_t^+(F_I)_j^{n} &= -\delta_t^+(F_I^\prime)_j^{n^\prime-1},
\end{aligned}
\end{equation}
for $I\in\{\text{s},\text{v}\}$.
Applying the time-reversal operator $\vK$ to the {\tt CN} scheme \eqref{eq:cn-gnld}
and using the commutation relation \eqref{k-relation} and Eqs.~\eqref{TR-e1} and \eqref{TR-e2}
lead to
\begin{equation*}
\begin{aligned}
\delta_t^+ {(\vPsi^\prime)}_j^{n^\prime-1}  + & \vgamma^0\vgamma^1 \ell_t^+\delta_x^0{(\vPsi^\prime)}_j^{n^\prime-1}
    + \mi m \vgamma^0 \ell_t^+ {(\vPsi^\prime)}_j^{n^\prime-1} - \mi
    \frac{\delta_t^+(F^\prime_\text{s})_j^{{n}^\prime-1}}
    {\delta_t^+(w^\prime_\text{s})_j^{{n}^\prime-1}}\vgamma^0\ell_t^+ {(\vPsi^\prime)}_j^{{n}^\prime-1} \\
    &-\mi \frac{\delta_t^+(F^\prime_\text{v})_j^{{n}^\prime-1}}
    {\delta_t^+(w^\prime_\text{v})_j^{{n}^\prime-1}}
    \ell_t^+(\wbar{\vPsi}^\prime\vgamma_\mu{\vPsi^\prime})_j^{{n}^\prime-1}\vgamma^0\vgamma^\mu \ell_t^+{(\vPsi^\prime)}_j^{n^\prime-1}
    = 0,
\end{aligned}
\end{equation*}
which is exactly the {\tt CN} scheme \eqref{eq:cn-gnld} applied to ${(\vPsi^\prime)}_j^{{n}^\prime-1}$.
That is, the {\tt CN} scheme  is invariant under the the time-reversal transformation,
namely, it is time reversible.

The fact that the {\tt LCN1} scheme \eqref{eq:lcn1-gnld} is not time reversible can be observed directly if noting
\begin{align*}
&-  \vK \mi \ell_t^\text{e}\left({f}_{\text{s}}\vgamma^0\vPsi
      + {f}_{\text{v}} \wbar{\vPsi}\vgamma_\mu\vPsi\vgamma^0\vgamma^\mu\vPsi \right)_j^n
      \\
=& \mi \left[\frac32 \left(f^\prime_\text{s}\vgamma^0 \vPsi^\prime +f^\prime_\text{v} \wbar{\vPsi}^\prime\vgamma_\mu \vPsi^\prime \vgamma^0\vgamma^\mu \vPsi^\prime \right)_j^{n^\prime}
  - \frac12 \left(f^\prime_\text{s}\vgamma^0 \vPsi^\prime + f^\prime_\text{v} \wbar{\vPsi}^\prime\vgamma_\mu \vPsi^\prime\vgamma^0\vgamma^\mu \vPsi^\prime \right)_j^{n^\prime+1}\right]
  \\
\neq & \mi \ell_t^\text{e} \left(f^\prime_\text{s}\vgamma^0 \vPsi^\prime + f^\prime_\text{v} \wbar{\vPsi}^\prime\vgamma_\mu
\vPsi^\prime\vgamma^0\vgamma^\mu \vPsi^\prime \right)_j^{n^\prime-1}.
\end{align*}
\end{proof}

Next,
we will discuss the conservation of the discrete energy, linear momentum and charge defined below
for the numerical methods given in Section \ref{sec:DiffSch}.
After performing the integration in the computational domain $\Omega = [x_\text{L},x_\text{R}]$
and then approximating the first derivative operator $\partial_x$
with the centered difference operator $\delta_x^0$ as well as the integral operator $\int_{x_\text{L}}^{x_\text{R}}\dif x$
with the summation
operator $h\sum_{j=1}^J$ in Eq.~\eqref{qep-infty},
we have the discrete energy, linear momentum and charge at the $n$-th time step
\begin{equation*}
  \begin{aligned}
  E_h^n &= h \sum_{j = 1}^J \left(\imag (\wbar{\vPsi}\vgamma^1\delta_x^0\vPsi)
    + m (\wbar{\vPsi} \vPsi)-F_{\text{s}}
    - \frac12 F_{\text{v}}\right)_j^n,\\
P_h^n &= h \sum_{j = 1}^J \imag (\vPsi^\dag \delta_x^0\vPsi)_j^n, \\
Q_h^n &= \| \vPsi^n\|^2 :=\left<\vPsi^n, \vPsi^n\right>= h \sum_{j = 1}^J (\vPsi_j^{n})^\dag \vPsi_j^{n},
\end{aligned}
\end{equation*}
where the inner product $\left< \cdot,\cdot \right>$ is defined as
\begin{equation*}
      \left<\vec u, \vec v\right> = h \sum_{j = 1}^J
      (\vec u_j)^\dag \vec v_j
\end{equation*}
for two complex-valued vectors $\vec u$ and $\vec v$, and $J$ is the grid point number.
Here the values of $\vPsi$ at $x_j$ with $1\leq j\leq J$ are unknowns and
those at $x_0$ and $x_{J+1}$ are determined by appropriate boundary conditions.

We first present the following lemma which can be verified through direct algebraic manipulations and will be used in
discussing the conservation of the discrete charge, energy and linear momentum.

\begin{Lemma}
  \label{le:1}
Given that $\vPhi_j^n$ is a complex-valued vector mesh function with two components evaluated at the mesh $\{x_j,t_n\}$
($n=0,1,\cdots,N$, $j=0,1,\cdots,J+1$) and a matrix $\vGamma\in \{\vI_2, \vgamma^0, \vgamma^1, \mi \vgamma^0\vgamma^1\}$,
we have the following identities

(a) $2\real\left<\ell_t^+\vPhi^n,\delta_t^+\vPhi^n\right> = \delta_t^+ \|\vPhi^n\|^2$;

(b) $2 \real \left<\vgamma^0\vgamma^1\delta_x^0\vPhi^n,\vPhi^n\right> =
\frac1{2}(\wbar{\vPhi}_{J+1}^n\vgamma^1\vPhi_J^n +
\wbar{\vPhi}_{J}^n\vgamma^1\vPhi_{J+1}^n)
-\frac1{2}(\wbar{\vPhi}_{1}^n\vgamma^1\vPhi_0^n +
\wbar{\vPhi}_{0}^n\vgamma^1\vPhi_{1}^n)$;

(c) $\imag\left(\ell_t^+\wbar{\vPhi}_j^n \vGamma \ell_t^+\vPhi_j^n \right) =0$;

(d) $2 \real (\delta_t^+\wbar{\vPhi}_j^n \vGamma \ell_t^+\vPhi_j^n)  =
      \delta_t^+(\wbar{\vPhi}_j^n\vGamma \vPhi_j^n)$;

(e) $
2 \mi \imag \left< \vgamma^0\vgamma^1 (\ell_t^+ \delta_x^0) \vPhi^n,\delta_t^+\vPhi^n\right> =
\begin{array}[t]{l}
-\delta_t^+\left(  h \sum_{j=1}^J (\wbar{\vPhi}_j^n\vgamma^1
        \delta_x^0\vPhi_j^n) \right) \\
        +\frac12\big(
        (\ell_t^+\wbar{\vPhi}_{J+1}^n\vgamma^1\delta_t^+\vPhi_J^n + \ell_t^+\wbar{\vPhi}_{J}^n\vgamma^1\delta_t^+\vPhi_{J+1}^n) \\
          -(\ell_t^+\wbar{\vPhi}_{1}^n\vgamma^1\delta_t^+\vPhi_0^n +
\ell_t^+\wbar{\vPhi}_{0}^n\vgamma^1\delta_t^+\vPhi_{1}^n)
\big).
\end{array}$
\end{Lemma}

\begin{proof}
It can be checked that the following Leibniz rules
\begin{equation*}
  \delta_t^a(\vec u^\dag\vec v)^n_j = \delta_t^a (\vec u^\dag)^n_j \ell_t^a \vec v^n_j + \ell_t^a (\vec u^\dag)^n_j\delta_t^a \vec v^n_j
\end{equation*}
holds for any two spinors $\vec u(x,t),\vec v(x,t)$ and $a\in\{+,-,0\}$,
and then we have
\begin{align*}
2\real \left<\ell_t^+\vPhi^n,\delta_t^+\vPhi^n\right>
&= \left<\delta_t^+\vPhi^n,\ell_t^+\vPhi^n\right> + \left<\ell_t^+\vPhi^n, \delta_t^+\vPhi^n\right> \\
&= h\sum_{j=1}^J \delta_t^+(\vPhi^\dag)_j^n \ell_t^+\vPhi_j^n + h\sum_{j=1}^J \ell_t^+(\vPhi^\dag)_j^n \delta_t^+\vPhi_j^n
\\
&= h\sum_{j=1}^J \delta_t^+ (\vPhi^\dag \vPhi)_j^n
= \delta_t^+ \|\vPhi^n\|^2.
\end{align*}
Thus the identity $(a)$ holds.

Because the operator $-\delta_x^0$ is the adjoint operator of $\delta_x^0$
and $\vgamma^0\vgamma^1$ is an Hermite matrix,
we get the identity $(b)$ directly by rearranging the summation.
The identity $(c)$ can be easily verified if using the fact $(\vgamma^0\vGamma)^\dag = \vgamma^0\vGamma$.
The proof of $(d)$ (resp. $(e)$) is similar with that of $(b)$ (resp. $(c)$).
\end{proof}

\begin{Proposition}
  \label{thm:conservation}
The {\tt CN}, {\tt CN0}, {\tt LCN0}, {\tt LCN2}, and {\tt OS($N$)} schemes conserve the discrete charge,
but only the {\tt CN} scheme conserves the discrete energy.
\end{Proposition}
\begin{proof} 
We begin with the discrete conservation law of charge for the {\tt CN} scheme \eqref{eq:cn-gnld}.
Performing the inner product of $\ell_t^+\Vec \Psi^n$ and the {\tt CN} scheme \eqref{eq:cn-gnld} leads to
\begin{equation}\label{cn-inner}
\begin{aligned}
\left<\ell_t^+\vPsi^n,\delta_t^+\vPsi^n\right> & +  \left<\ell_t^+\vPsi^n,\vgamma^0\vgamma^1 \delta_x^0\ell_t^+\vPsi^n\right>
\\
    & + \left<\ell_t^+\vPsi^n,\mi m \vgamma^0 \ell_t^+\vPsi^{n} - \mi
    \frac{\delta_t^+(F_{\text{s}})^n}
    {\delta_t^+(w_{\text{s}})^n}\vgamma^0\ell_t^+\vPsi^n\right>  \\
    & +\left<\ell_t^+\vPsi^n,-\mi \frac{\delta_t^+(F_{\text{v}})^n}
    {\delta_t^+(w_{\text{v}})^n}
    \ell_t^+(\wbar{\vPsi}\vgamma_\mu\vPsi)^n\vgamma^0\vgamma^\mu \ell_t^+\vPsi^n\right>
    = 0,
\end{aligned}
\end{equation}
and then the
conservation law of the discrete charge can be easily verified by
taking directly the real part as follows
\begin{equation}\label{eq:dis-Q-CN}
\delta_t^+ Q_h^n=
\frac1{2}(\ell_t^+\wbar{\vPsi}_{1}^n\vgamma^1\ell_t^+\vPsi_0^n +
\ell_t^+\wbar{\vPsi}_{0}^n\vgamma^1\ell_t^+\vPsi_{1}^n)
- \frac1{2}(\ell_t^+\wbar{\vPsi}_{J+1}^n\vgamma^1\ell_t^+\vPsi_J^n +
 \ell_t^+\wbar{\vPsi}_{J}^n\vgamma^1\ell_t^+\vPsi_{J+1}^n),
\end{equation}
where Lemma \ref{le:1} (a) is applied to the first term in Eq.~\eqref{cn-inner},
(b) to the second term and (c) to the third and fourth terms.
Similarly, it can be verified that \eqref{eq:dis-Q-CN} holds
for the {\tt CN0}~\eqref{eq:cn-kuo}, {\tt LCN0}~\eqref{eq:lcn-kuo} and {\tt LCN2} \eqref{eq:lcn2-gnld} schemes.

Performing the inner product of the {\tt CN} scheme \eqref{eq:cn-gnld} and $\delta_t^+\vPsi^n$,
keeping the imaginary part and applying Lemma \ref{le:1} give directly
the conservation of the discrete energy
\begin{align*}
  \delta_t^+ E_h^n = \frac{1}2 \imag \left((\ell_t^+\wbar{\vPsi}_{J+1}^n\vgamma^1\delta_t^+\vPsi_J^n +
\ell_t^+\wbar{\vPsi}_{J}^n\vgamma^1\delta_t^+\vPsi_{J+1}^n)\right. \\
\left.
-(\ell_t^+\wbar{\vPsi}_{1}^n\vgamma^1\delta_t^+\vPsi_0^n +
\ell_t^+\wbar{\vPsi}_{0}^n\vgamma^1\delta_t^+\vPsi_{1}^n)\right).
\end{align*}

For the Fourier spectral method~\eqref{eq:fft},
using the fact that  $\exp{(-\mi \tau_i (\kappa \vgamma^0\vgamma^1+m\vgamma^0))}$
in \eqref{eq:matrixfft} is a unitary matrix
yields
\begin{equation*}
\left\|\exp\left( -\mi \tau_i (\kappa \vgamma^0\vgamma^1+m\vgamma^0)\right) \vPsi_j^{(1)}\right\|^2 = \left\|\vPsi_j^{(0)}\right\|^2,
\end{equation*}
and then
\begin{align*}
  \left\|\vPsi^{(1)}\right\|^2 &= \left\|\mathcal{F}^{-1}
  \left[\exp\left( -\mi \tau_i (\kappa \vgamma^0\vgamma^1+m\vgamma^0)\right)
    \mathcal{F}\left(\vPsi^{(0)}\right)\right]\right\|^2 \\
  &= \left\|\exp\left( -\mi \tau_i (\kappa \vgamma^0\vgamma^1+m\vgamma^0)\right)
    \mathcal{F}\left(\vPsi^{(0)}\right)\right\|^2 = \left\| \mathcal{F}\left(\vPsi^{(0)}\right)\right\|^2
 = \left\|\vPsi^{(0)}\right\|^2,   
\end{align*}
where Parseval's identity is applied twice.
It can be readily verified that the matrix exponents in Eqs.~\eqref{OS-NS-eq3} and \eqref{eq:solution-V}
are unitary, thus $\|\vPsi^{(1)}\|^2=\|\vPsi^{(0)}\|^2$ holds both for
Eqs.~\eqref{OS-NS-eq3} and \eqref{eq:solution-V}, \ie
$Q_h$ should be conserved for solutions of the nonlinear subproblems.
Therefore, the {\tt OS($N$)} scheme satisfies the conservation law of charge.
\end{proof}

\begin{Remark}\rm
It will be verified by numerical results in Section~\ref{sec:numericalresults} that
the {\tt LCN1}, {\tt HS}, {\tt LF} and {\tt SI} schemes do not conserve the discrete charge or energy,
and none of the numerical methods presented in Section \ref{sec:DiffSch} conserves the discrete linear momentum.
\end{Remark}

\section{Numerical results}
\label{sec:numericalresults}

This section will conduct numerical simulations to compare
the performance of numerical schemes proposed in Section~\ref{sec:DiffSch}
and then utilize the {\tt OS($4$)} scheme to investigate the
interaction dynamics for the NLD solitary waves  \eqref{mw} under the
scalar and vector self-interaction.
For those localized NLD solitary waves,
the periodic boundary condition for the {\tt OS($N$)} scheme and the non-reflection boundary condition
for other schemes could be adopted at the boundaries of the computational domain if a relatively large computational domain
 has been taken in our numerical experiments.

All calculations are performed on a Lenovo desktop computer with
Intel Core i5 650 CPU and 4GB RAM using double precision in the
3.0.0-24-generic x86\_64 Linux operation system and the compiler is gcc 4.6.1.
The computational domain $\Omega$ will be taken as $[-50,50]$ in Examples \ref{eg.1}-\ref{eg.scalar-vector}
and  $[-100,100]$ in Example \ref{eg.scalar-2humped}.
and the particle mass $m$ in Eq.~\eqref{nld-lorentz} is chosen to be $1$.

\begin{Example}
\label{eg.1}
The first example is devoted to comparing the numerical performance of all the numerical methods in Section \ref{sec:DiffSch}
in terms of the accuracy, the conservativeness, the efficiency and the error growth.
A one-humped  solitary wave with the velocity $V=-0.2$ is simulated here under the quadric scalar self-interaction (\ie $v=0$ and $k=1$), traveling from right to left with the parameters in \eqref{mw}:
$x_0 = 5$, $s= 0.5$, and $\omega = 0.75$.
The $P^N$-RKDG method  \cite{ShaoTang2006} is also included here for comparison,
which is assembled with a fourth-order accurate Runge-Kutta time discretization  in time
and the Legendre polynomials of degree at most $N$ as local basis functions in the spatial
Galerkin approximation.

\begin{table}
  \centering
  \caption{Example~\ref{eg.1}. Part I: Numerical comparison of the accuracy, the conservativeness
  and the efficiency at $t=50$  with the time stepsize being set to $\tau = \frac12 h$.
  The CPU time in seconds is recorded for the finest mesh.}
  \label{tab:eg1}
 \newsavebox{\tablebox}
 \begin{lrbox}{\tablebox}
  \begin{tabular}{lrcccccccc}
&&&&&&&&\\
\hline                  &       $h$     &       $\mathcal{V}_{Q}$ &       $\mathcal{V}_{E}$ &       $\mathcal{V}_{P}$ &       $\text{err}_2$      &       Order   &       $\text{err}_\infty$        &       Order   &
Time (s)   \\ \hline
{\tt CN}                      &       0.08    &       1.38E-15        &       1.09E-15        &       2.65E-06        &       2.74E-02        &               &       2.61E-02        &               &               \\
                        &       0.04    &       4.66E-15        &       1.37E-15        &       1.65E-07        &       6.84E-03        &       2.00    &       6.50E-03        &       2.00    &               \\
                        &       0.02    &       4.78E-15        &       1.37E-15        &       1.03E-08        &       1.71E-03        &       2.00    &       1.63E-03        &       2.00    &               \\
                        &       0.01    &       3.05E-14        &       2.05E-15        &       6.42E-10        &       4.27E-04        &       2.00    &       4.06E-04        &       2.00    &       367.1  \\ \hline
{\tt CN0}                     &       0.08    &       7.05E-15        &       3.24E-08        &       2.37E-06        &       2.31E-02        &               &       2.21E-02        &               &               \\
                        &       0.04    &       1.01E-15        &       2.02E-09        &       1.47E-07        &       5.76E-03        &       2.00    &       5.51E-03        &       2.00    &               \\
                        &       0.02    &       5.54E-15        &       1.26E-10        &       9.17E-09        &       1.44E-03        &       2.00    &       1.38E-03        &       2.00    &               \\
                        &       0.01    &       2.88E-14        &       7.88E-12        &       5.72E-10        &       3.59E-04        &       2.00    &       3.44E-04        &       2.00    &       345.3  \\ \hline
{\tt LCN0}                    &       0.08    &       2.14E-15        &       1.83E-06        &       4.31E-05        &       2.75E-02        &               &       2.62E-02        &               &               \\
                        &       0.04    &       1.26E-16        &       2.29E-07        &       5.55E-06        &       6.87E-03        &       2.00    &       6.53E-03        &       2.00    &               \\
                        &       0.02    &       8.81E-15        &       2.86E-08        &       7.04E-07        &       1.72E-03        &       2.00    &       1.63E-03        &       2.00    &               \\
                        &       0.01    &       3.11E-14        &       3.58E-09        &       8.87E-08        &       4.29E-04        &       2.00    &       4.08E-04        &       2.00    &       81.8 \\ \hline
{\tt LCN1}                    &       0.08    &       2.22E-04        &       1.92E-04        &       3.73E-04        &       3.53E-02        &               &       3.34E-02        &               &               \\
                        &       0.04    &       2.78E-05        &       2.40E-05        &       4.65E-05        &       9.06E-03        &       1.96    &       8.56E-03        &       1.97    &               \\
                        &       0.02    &       3.47E-06        &       3.01E-06        &       5.80E-06        &       2.30E-03        &       1.98    &       2.17E-03        &       1.98    &               \\
                        &       0.01    &       4.34E-07        &       3.76E-07        &       7.25E-07        &       5.78E-04        &       1.99    &       5.45E-04        &       1.99    &       118.7  \\ \hline
{\tt LCN2}                    &       0.08    &       2.77E-15        &       1.64E-07        &       7.01E-06        &       2.77E-02        &               &       2.64E-02        &               &               \\
                        &       0.04    &       7.55E-16        &       1.98E-08        &       6.77E-07        &       6.92E-03        &       2.00    &       6.58E-03        &       2.00    &               \\
                        &       0.02    &       7.55E-16        &       2.44E-09        &       7.24E-08        &       1.73E-03        &       2.00    &       1.64E-03        &       2.00    &               \\
                        &       0.01    &       2.91E-14        &       3.03E-10        &       8.29E-09        &       4.32E-04        &       2.00    &       4.11E-04        &       2.00    &       118.4  \\ \hline
{\tt HS}                      &       0.08    &       1.62E-06        &       1.55E-06        &       4.33E-06        &       2.00E-02        &               &       1.55E-02        &               &               \\
                        &       0.04    &       1.01E-07        &       9.65E-08        &       2.70E-07        &       5.00E-03        &       2.00    &       3.88E-03        &       2.00    &               \\
                        &       0.02    &       6.30E-09        &       6.03E-09        &       1.68E-08        &       1.25E-03        &       2.00    &       9.71E-04        &       2.00    &               \\
                        &       0.01    &       3.94E-10        &       3.77E-10        &       1.05E-09        &       3.13E-04        &       2.00    &       2.43E-04        &       2.00    &       14.8   \\ \hline
{\tt LF}                      &       0.08    &       5.88E-05        &       4.73E-05        &       8.68E-06        &       1.41E-02        &               &       1.35E-02        &               &               \\
                        &       0.04    &       5.51E-06        &       4.43E-06        &       6.59E-07        &       3.51E-03        &       2.01    &       3.36E-03        &       2.01    &               \\
                        &       0.02    &       6.24E-07        &       5.01E-07        &       6.91E-08        &       8.74E-04        &       2.00    &       8.36E-04        &       2.00    &               \\
                        &       0.01    &       7.54E-08        &       6.05E-08        &       8.19E-09        &       2.18E-04        &       2.00    &       2.09E-04        &       2.00    &       8.9    \\ \hline
{\tt SI}                      &       0.08    &       3.79E-08        &       1.15E-07        &       3.94E-06        &       3.59E-02        &               &       4.76E-02        &               &               \\
                        &       0.04    &       3.68E-09        &       8.36E-09        &       2.40E-07        &       8.97E-03        &       2.00    &       1.19E-02        &       2.00    &               \\
                        &       0.02    &       4.09E-10        &       6.87E-10        &       1.44E-08        &       2.24E-03        &       2.00    &       2.98E-03        &       2.00    &               \\
                        &       0.01    &       4.89E-11        &       6.46E-11        &       8.20E-10        &       5.60E-04        &       2.00    &       7.44E-04        &       2.00    &       132.7  \\ \hline
  \end{tabular}
\end{lrbox}
\scalebox{0.7}{\usebox{\tablebox}}
\end{table}

\begin{table}
  \centering
  \caption{Example~\ref{eg.1}. Part II:
  Numerical comparison of the accuracy, the conservativeness
  and the efficiency at $t=50$.
  The CPU time measured in seconds is listed for the finest mesh.}
  \label{tab:eg1-2}
\begin{lrbox}{\tablebox}
  \begin{tabular}{lrrcccccccc}
&&&&&&&&&&\\
\hline  &       $\tau$  &       $h$     &       $\mathcal{V}_{Q}$ &       $\mathcal{V}_{E}$ &       $\mathcal{V}_{P}$ &       $\text{err}_{2}$     &       Order   &       $\text{err}_{\infty}$        &       Order   &       Time (s)    \\ \hline
{\tt OS($2$)} &       0.04    &       0.78    &       1.04E-13        &       2.33E-07        &       3.97E-05        &       2.03E-03        &               &       9.96E-04        &               &               \\
        &       0.02    &       0.39    &       1.33E-13        &       1.45E-08        &       2.55E-13        &       2.53E-04        &       3.01    &       1.90E-04        &       2.39    &               \\
        &       0.01    &       0.39    &       3.58E-13        &       9.04E-10        &       6.08E-13        &       6.33E-05        &       2.00    &       4.75E-05        &       2.00    &               \\
        &       0.005   &       0.20    &       6.15E-13        &       5.60E-11        &       7.13E-13        &       1.58E-05        &       2.00    &       1.21E-05        &       1.97    &       8.2    \\ \hline
{\tt OS($4$)} &       0.04    &       0.78    &       7.40E-13        &       5.78E-13        &       4.12E-05        &       1.71E-03        &               &       4.27E-04        &               &               \\
        &       0.02    &       0.39    &       1.46E-12        &       1.25E-12        &       2.17E-12        &       7.92E-08        &       14.40   &       3.32E-08        &       13.65   &               \\
        &       0.01    &       0.20    &       1.88E-12        &       1.60E-12        &       2.71E-12        &       4.28E-10        &       7.53    &       3.28E-10        &       6.66    &               \\
        &       0.005   &       0.10    &       7.00E-13        &       6.55E-13        &       2.44E-12        &       1.93E-11        &       4.47    &       1.44E-11        &       4.51    &       16.8   \\ \hline
$P^3$-RKDG      &       0.04    &       0.78    &       1.03E-05        &       5.47E-07        &       2.61E-06        &       1.46E-04        &               &       1.44E-04        &               &               \\
        &       0.02    &       0.39    &       8.60E-08        &       5.96E-07        &       6.31E-07        &       6.59E-06        &       4.47    &       8.85E-06        &       4.03    &               \\
        &       0.01    &       0.20    &       6.98E-10        &       4.32E-08        &       4.69E-08        &       4.10E-07        &       4.01    &       5.62E-07        &       3.98    &               \\
        &       0.005   &       0.10    &       7.88E-12        &       2.75E-09        &       3.00E-09        &       2.56E-08        &       4.00    &       3.69E-08        &       3.93    &       551.6  \\ \hline
$P^4$-RKDG      &       0.04    &       0.78    &       1.00E-07        &       2.13E-07        &       1.10E-07        &       8.02E-06        &               &       1.04E-05        &               &               \\
        &       0.02    &       0.39    &       8.53E-10        &       4.51E-09        &       3.44E-09        &       2.58E-07        &       4.96    &       3.36E-07        &       4.96    &               \\
        &       0.01    &       0.20    &       2.28E-11        &       6.41E-11        &       5.22E-11        &       8.46E-09        &       4.93    &       1.16E-08        &       4.86    &               \\
        &       0.005   &       0.10    &       2.58E-12        &       9.56E-11        &       2.16E-10        &       3.09E-10        &       4.78    &       2.39E-10        &       5.60    &       744.8  \\ \hline
\end{tabular}
\end{lrbox}
\scalebox{0.7}{\usebox{\tablebox}}
\end{table}

Tables \ref{tab:eg1} and \ref{tab:eg1-2} summarize the numerical
results at the final time $t=50$, where
$\text{err}_2$ and $\text{err}_\infty$ are respectively
the $l^2$ and $l^\infty$ errors at the final time,
$\mathcal{V}_Q$, $\mathcal{V}_E$, $\mathcal{V}_P$
measure respectively the variation of charge, energy and linear momentum
at the final time relative to the initial quantities,
and the CPU time of calculations
with the same mesh size is recorded for comparing the efficiency.
It can be observed clearly there that:
(1) while the {\tt CN}, {\tt CN0}, {\tt LCN0}, {\tt LCN1}, {\tt LCN2}, {\tt HS}, {\tt LF}, {\tt SI} and {\tt OS($2$)} schemes
are of the second-order accuracy, the {\tt OS($4$)}, $P^3$-RKDG and $P^4$-RKDG methods
exhibit at least the fourth-order accuracy;
(2) The {\tt CN}, {\tt CN0}, {\tt LCN0}, {\tt LCN2}, {\tt OS($2$)} and {\tt OS($4$)} schemes conserve
the discrete charge and only the {\tt CN} scheme conserves the discrete energy,
but none conserves the discrete linear momentum;
(3) The {\tt OS($4$)} scheme could also keep very accurately the discrete energy and linear momentum with
relatively fine meshes.
All above numerical results are consistent with the theoretical results given in Section~\ref{sec:discussion4fdm}.
Among the numerical methods of the second order accuracy,
it is also found that the {\tt OS($2$)} scheme runs fastest ($8.2$ seconds for the mesh $\tau=0.005$ and $h=0.20$)
if requiring to attain almost the same accuracy.
Similarly, the {\tt OS($4$)} scheme runs much more faster than both $P^3$-RKDG and $P^4$-RKDG methods,
and the ratio of the CPU time used by the {\tt OS($4$)} scheme over that used by the $P^3$-RKDG method is around $3.05\%$,
and reduces to around $2.26\%$ over that used by the $P^4$-RKDG method.

\begin{figure}[h]
\centering
\includegraphics[width=6.5cm]{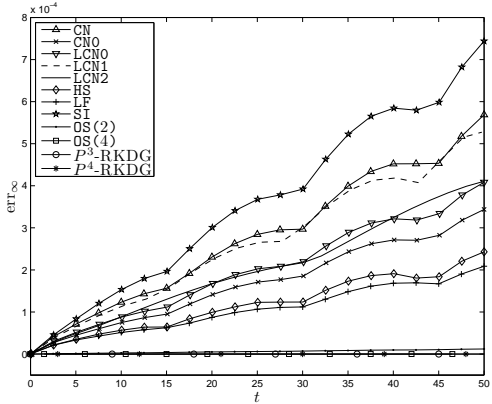}
\caption{Example~\ref{eg.1}. The $l^\infty$ error history.}
\label{fig.LinfErr}
\end{figure}

Fig.~\ref{fig.LinfErr} plots the $l^\infty$ error history
in the finest mesh used in Tables \ref{tab:eg1} and \ref{tab:eg1-2}.
According to the curves shown there,
it can be seen there that the $l^\infty$ error of all the schemes increases almost linearly with the time.
However, the slopes, obtained by the linear fitting, are different.
The smaller the slope is, the longer time the scheme could simulate to.
The {\tt SI} scheme has the largest slope $1.412\times 10^{-05}$ while the {\tt OS($2$)} scheme has
the smallest one $2.334\times 10^{-07}$ among all the second-order accurate methods.
Further, the slopes of the curves of $l^\infty$ errors for
the {\tt OS($4$)}, $P^3$-RKDG and $P^4$-RKDG schemes are almost the same
value of $3.199\times 10^{-13}$ which is much more smaller
than those of the second-order accurate schemes.

\begin{Remark}\rm
Both the theoretical and numerical comparison of the {\tt OS($2$)} scheme with the {\tt CN0} and {\tt LCN0} schemes
show that the former is better, especially in terms of efficiency and error growth.
Therefore in some sense this is an answer to the debate stimulated in \cite{Alvarez1992, FrutosSanz-Serna1989}
over twenty years ago on which one is most efficient among
the {\tt OS($2$)}, {\tt CN0} and {\tt LCN0} schemes.
\end{Remark}
\end{Example} 

\begin{Example}
\label{eg.os-dg}

\begin{figure}[h]
\centering
  \includegraphics[width=6.5cm]{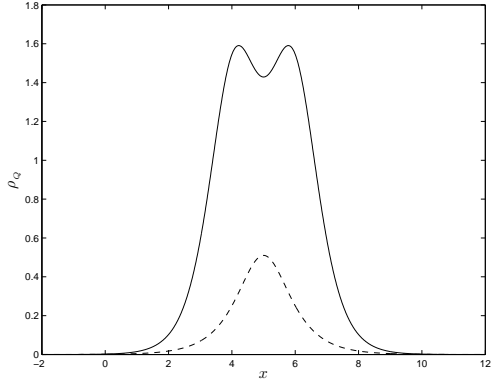}
  \caption{Example~\ref{eg.os-dg}. The initial charge density $\rho_{Q}(x,t)$ for two typical cases,
  $V =-0.2$, $x_0=5$.
  Case $1$ is shown in the solid line, a two-humped profile ($\omega = 0.3$) under the quadric scalar self-interaction ($k=1$, $s=0.5$ and $v=0$); Case $2$ is shown in the dashed line, a one-humped profile ($\omega = 0.75$)   under the cubic scalar and vector self-interaction ($k=2$ and $s=v=0.5$)}
\label{fig.os-dg-rhoq}
\end{figure}

The $P^3$-RKDG method
has been successfully applied before into investigating the interaction for the
NLD solitary waves under the quadric scalar self-interaction in
\cite{ShaoTang2005,ShaoTang2006,ShaoTang2008},
but the numerical comparison shown in Example~\ref{eg.1} tells us that
the proposed {\tt OS($4$)} scheme should be preferred now.
In this example, we further conduct numerical comparison among the {\tt OS($4$)},
$P^3$-RKDG and $P^4$-RKDG methods in simulating one-humped and two-humped solitary waves.
Two typical profiles of the charge density for the NLD solitary wave
displayed in Fig.~\ref{fig.os-dg-rhoq} are considered,
one denoted by Case $1$ has a two-humped profile under the quadric scalar self-interaction,
and the other denoted by Case $2$ has a one-humped profile under the cubic scalar and vector self-interaction.
These two solitary waves are located initially at $x_0=5$,
travel from right to left with the velocity $V =-0.2$ and stop at the final time $t=50$.
For convenience, we use $\rho_Q(x,t)\equiv J_0$ to
represent the charge density.

\begin{table}[h]
  \centering
  \caption{Example~\ref{eg.os-dg}. Numerical comparison among the {\tt OS($4$)},
$P^3$-RKDG and $P^4$-RKDG methods. The CPU time is measured in seconds. }
  \label{tab:os2-4th-vs-rkdg}
  Case 1 in the mesh of $\tau = 0.01$ and $h=\frac{100}{512}$
\begin{lrbox}{\tablebox}
  \begin{tabular}{ccccccc}
\hline\hline  &       $\mathcal{V}_{Q}$ &       $\mathcal{V}_{E}$ &       $\mathcal{V}_{P}$ &       $\text{err}_{2}$     &       $\text{err}_{\infty}$        &       Time (s)    \\ \hline
{\tt OS($4$)} &       1.96E-12       &         1.10E-12        &        3.45E-12        &       2.15E-09        &       2.17E-09        &       7.8    \\
$P^3$-RKDG      &       3.34E-08        &       2.27E-07        &       2.96E-07        &       2.89E-06        &       4.15E-06        &       146.6  \\
$P^4$-RKDG &       1.35E-12        &       9.43E-08        &       5.34E-08        &       9.02E-08        &       6.94E-08        &       195.5  \\
\hline
\end{tabular}
\end{lrbox}
\scalebox{0.7}{\usebox{\tablebox}}

Case 2 in the mesh of $\tau = 0.005$ and $h=\frac{100}{1024}$
\begin{lrbox}{\tablebox}
\begin{tabular}{ccccccc}
\hline\hline  &       $\mathcal{V}_{Q}$ &       $\mathcal{V}_{E}$ &       $\mathcal{V}_{P}$ &       $\text{err}_{2}$     &       $\text{err}_{\infty}$        &       Time (s)    \\ \hline
{\tt OS($4$)} &       9.82E-13        &       8.35E-13        &       2.95E-12        &       7.33E-11        &       8.63E-11        &       46.8   \\
$P^3$-RKDG      &       3.21E-10        &       9.04E-09        &       1.95E-08        &       2.00E-07        &       3.98E-07        &       881.1 \\
$P^4$-RKDG      &       9.10E-13        &       4.89E-09        &       1.57E-09        &       4.68E-09        &       6.35E-09        &       1156.5        \\
 \hline
  \end{tabular}
\end{lrbox}
\scalebox{0.7}{\usebox{\tablebox}}
\end{table}

The numerical comparison is shown in
Table~\ref{tab:os2-4th-vs-rkdg}, from which we can observe that,
(1) with the same mesh, no matter sparse ($\tau = 0.01$ and $h=\frac{100}{512}\approx0.1953$) or fine
($\tau = 0.005$ and $h=\frac{100}{1024}\approx0.0977$),
the {\tt OS($4$)} scheme is more conservative and higher accurate than both $P^3$-RKDG and $P^4$-RKDG  methods;
(2) the {\tt OS($4$)} scheme runs much faster than both $P^3$-RKDG and $P^4$-RKDG  methods as we have found in Table~\ref{tab:eg1-2}. Here,
the ratio of the CPU time used by the {\tt OS($4$)} scheme over that used by the $P^3$-RKDG method is around $5.32\%$,
and reduces to around $4.05\%$ over that used by the $P^4$-RKDG method for both cases.
The $l^\infty$ error history is plotted in Fig.~\ref{fig.os-dg},
which shows that those three methods all have almost zero slope (under $4.50$E-$09$).
This can be used to explain our previous success of the $P^3$-RKDG method in
\cite{ShaoTang2005,ShaoTang2006,ShaoTang2008}.
Further more, the more smaller errors of the {\tt OS($4$)} method mean that it should be more powerful
than others.

\begin{figure}[h]
\centering
  \includegraphics[width=6.5cm]{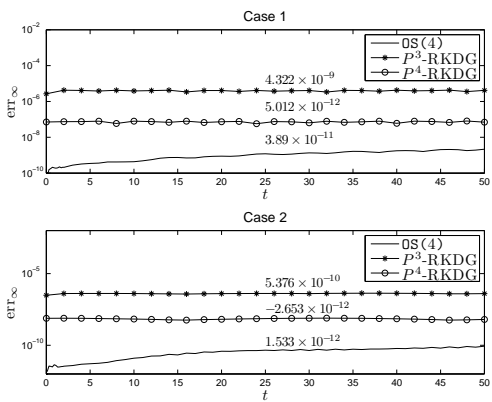}
\caption{Example~\ref{eg.os-dg}. The $l^\infty$ error history. The slopes are displayed above the curves.}
\label{fig.os-dg}
\end{figure}
\end{Example} 

It has been shown that the {\tt OS($4$)} scheme behaves best
for both one-humped and two-humped NLD solitary waves in long time simulations.
Therefore,
we conclude the comparison with the judgement that the {\tt OS($4$)} scheme is the most suitable for
simulating the interaction dynamics for the NLD solitary waves
in terms of the accuracy, the conservativeness, the efficiency and the error growth.
The {\tt OS($4$)} scheme will be utilized to investigate the binary collision of the NLD solitary waves.
The initial setup is the linear superposition of two moving waves
\begin{equation}
  \label{eq:two-mw-wave}
  \vPsi(x,t=0) = \vPsi_\text{l}^{\text{mw}}(x-x_{\text{l}},t=0)+\vPsi_\text{r}^{\text{mw}}(x-x_\text{r},t=0),
\end{equation}
where $\vPsi_{pos}^\text{mw}(x-x_{pos},t)$ denote the moving waves \eqref{mw} centered at $x_{pos}$ with the speed $V_{pos}$ and the frequency $\omega_{pos}$ for $pos \in\{\text{l},\text{r}\}$.
In the following examples, two equal solitary waves are placed symmetrically at $t=0$ with
$-x_\text{l}=x_\text{r}=10$ and $V_\text{l} = -V_\text{r} = 0.2$.
Several typical NLD solitary waves are considered with the parameters given in Table~\ref{3caselist},
and both quadric ($k=1$) and cubic ($k=2$) cases will be studied.
It should be noted that
the experiments carried out in the literatures
are all limited to the collisions of
the NLD solitary waves under the quadric scalar self-interaction.
A relatively fine mesh,  $\tau=0.005$ and $h=100/2^{13}\approx0.0122$, is adopted hereafter.

\begin{table}[h]
\caption{The initial setups of different cases in binary collisions.}
\label{3caselist}
\centering
\begin{tabular}{cccc|l}
\hline
case & $s$ & $v$ & $\omega_l=\omega_r$ & Remarks \\ \hline
B1   & 0.5 & 0   & 0.8                 & scalar, one-humped  \\
B2   & 0   & 0.5 & 0.8                 & vector, one-humped  \\
B3   & 0.5 & 0.5 & 0.8                 & scalar and vector, one-humped  \\
B4   & 0.5 & 0   & 0.3                 & scalar, two-humped  \\
B5   & 4   & 1   & 0.1                 & scalar and vector, two-humped   \\
\hline
\end{tabular}
\end{table}

\begin{Example}
\label{eg.scalar}

\begin{figure}[h]
  \centering
  \subfigure[$k=1$]{
  \includegraphics[width=6.5cm]{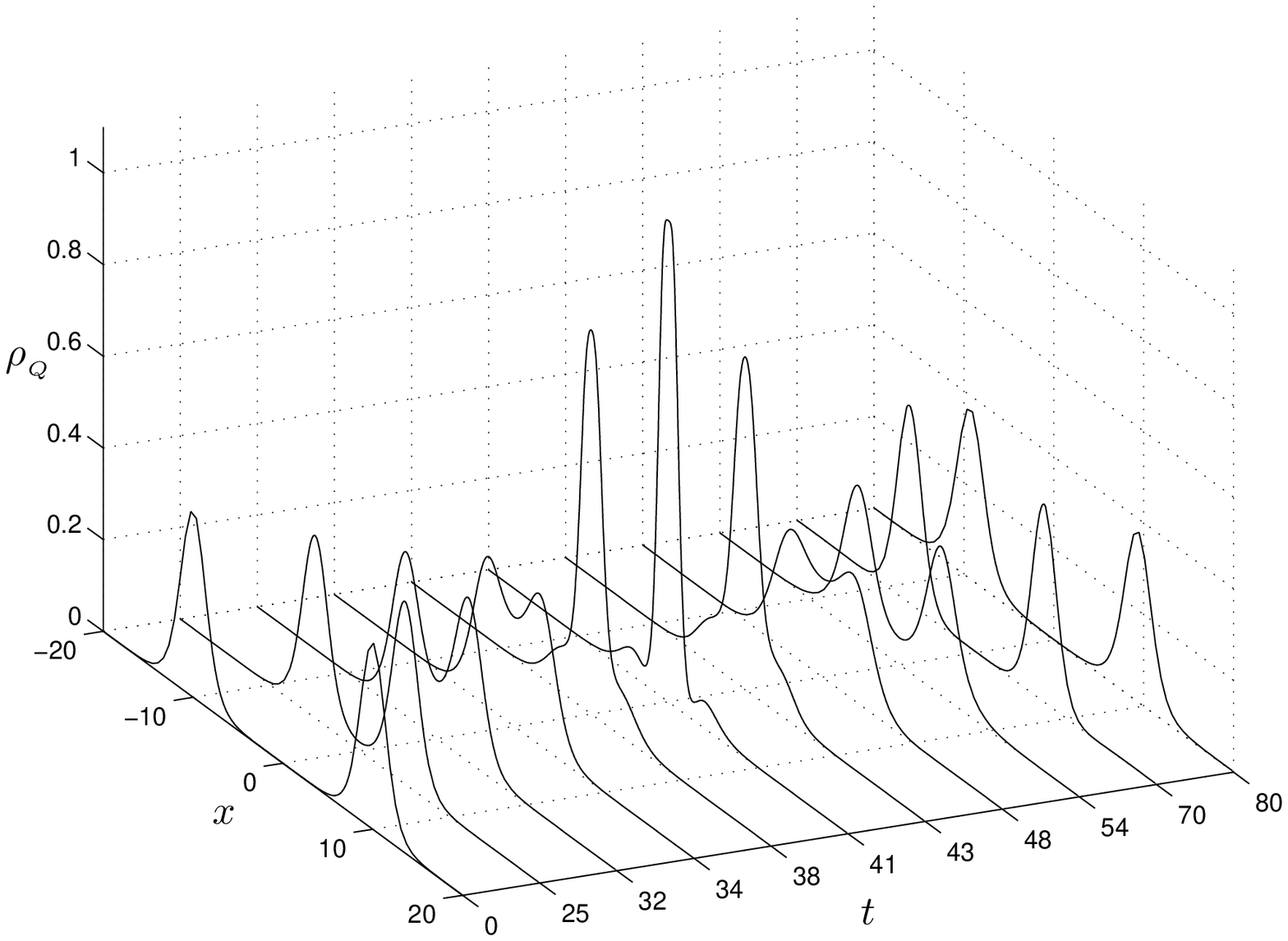}
  \label{subfigure:k1-s-rho}}
  \subfigure[$k=2$]{
  \label{subfigure:k1-s-QEP} 
  \includegraphics[width=6.5cm]{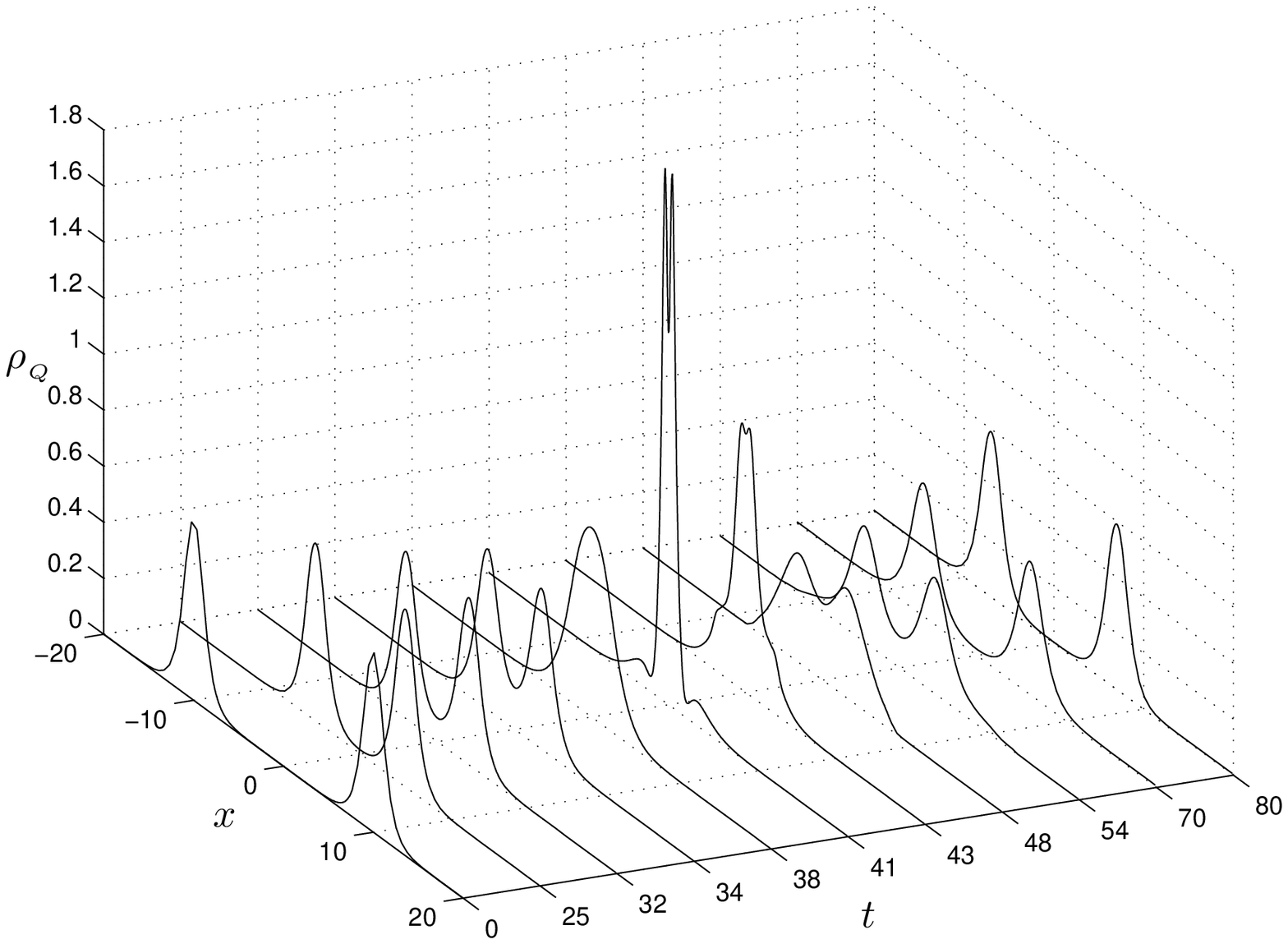}}
  \caption{Example \ref{eg.scalar}. Binary collision of the NLD solitary waves under the scalar self-interaction.}
  \label{fig:k1-s}
\end{figure}

The collision of two equal one-humped solitary waves under the scalar self-interaction,
\ie Case B1 in Table~\ref{3caselist}, is studied in this example. 
The interaction dynamics for the quadric case are shown in the left plot of Fig.~\ref{fig:k1-s},
where two equal waves with the initial amplitude of $0.4082$ move close at a velocity of $0.2$ and overlap each other,
then separate into a left moving wave and a right moving wave with the amplitude of $0.3743$ and the velocity of $0.1831$.
Similar phenomena are observed for the cubic case shown in the right plot of Fig.~\ref{fig:k1-s} except that
(1) two waves overlap more stronger around $t=41$ now due to the stronger nonlinearity;
(2) after collision, the amplitude decreases to $0.5899$ from the initial amplitude of $0.6455$ while
the velocity also decreases to $0.1037$.
In both cases, the discrete charge, energy and linear momentum are approximately conserved in the interaction
since the variation of them at $t=80$ is under $1.58$\text{E}-$10$.
\end{Example} 

\begin{Example}
\label{eg.v}

\begin{figure}[h]
  \centering
  \subfigure[$k=1$]{
  \includegraphics[width=6.5cm]{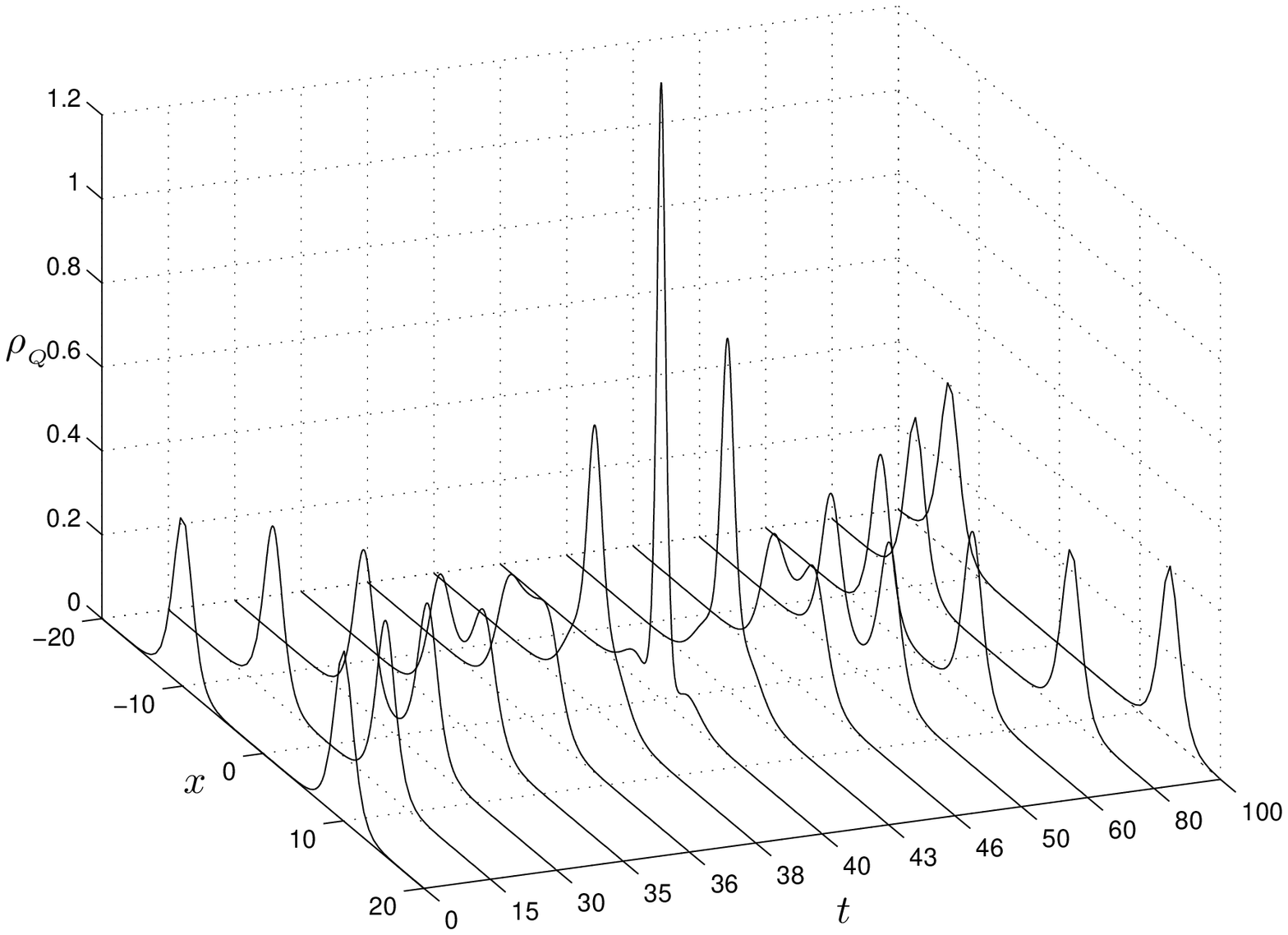}
  \label{subfigure:k1-v-rho}}
  \subfigure[$k=2$]{
  \label{subfigure:k1-v-QEP} 
  \includegraphics[width=6.5cm]{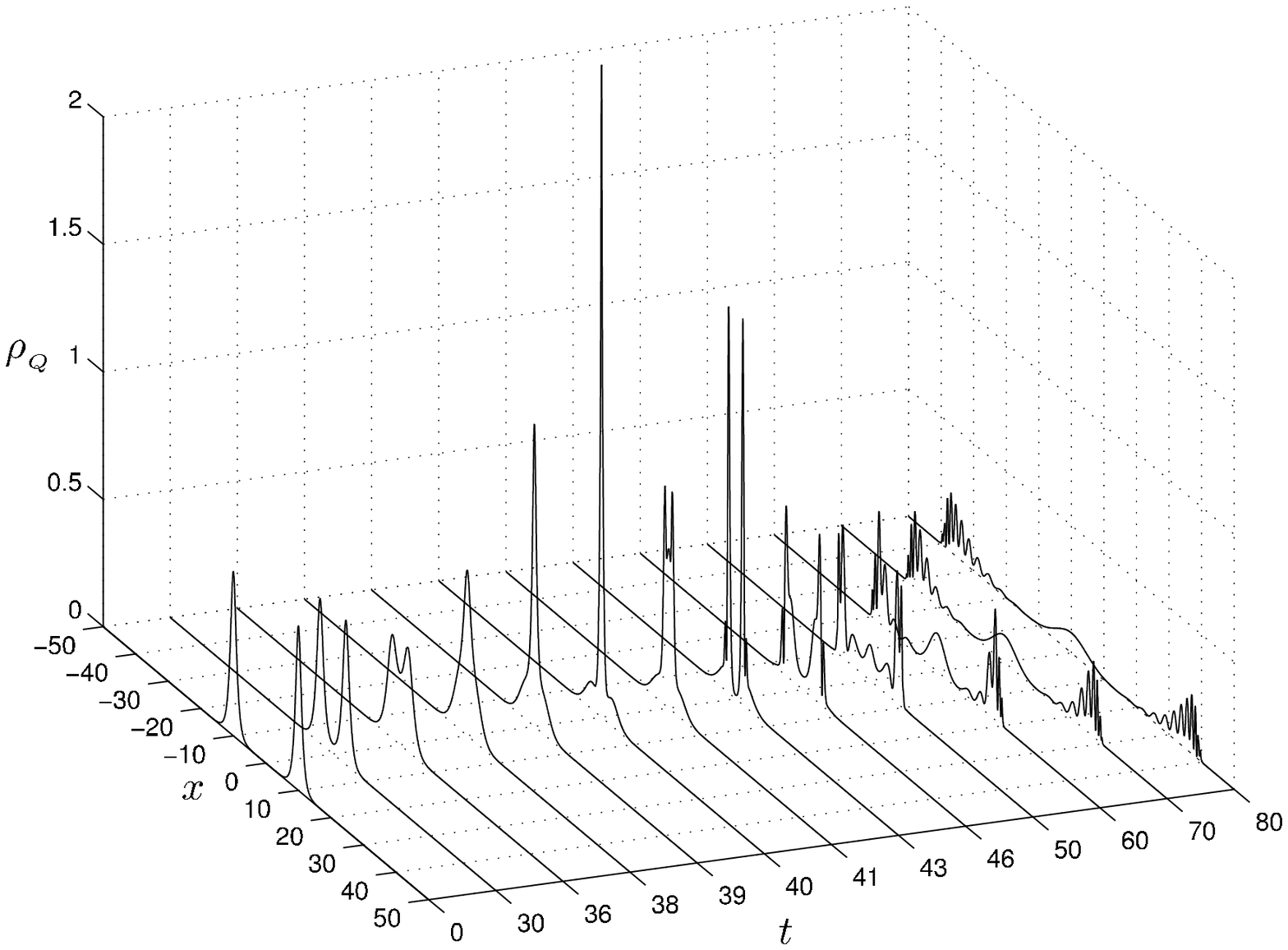}}
  \caption{Example \ref{eg.v}: Binary collision of the NLD solitary waves under the vector self-interaction.}
  \label{fig:k1-v}
\end{figure}

The collision of two equal one-humped solitary waves under the vector self-interaction,
\ie Case B2 in Table~\ref{3caselist}, is studied in this example. 
To the best of our knowledge, it is the first time to study binary collision of
the NLD solitary waves under the vector self-interaction.
The interaction dynamics for the quadric case are shown in the left plot of Fig.~\ref{fig:k1-v},
where the waves keep the shape and the velocity after the collision.
A totally different phenomenon appears for the cubic vector self-interaction as
displayed in the right plot of Fig.~\ref{fig:k1-v}.
The initial one-humped equal waves first merge into a single wave, then separate and overlap again.
Around $t=50$, collapse happens and highly oscillatory waves are generated and moving outside with a big velocity near $1$,
meanwhile a one-humped wave with small amplitude is formed at the center.
In both cases, the discrete charge, energy and linear momentum are approximately conserved in the interaction
since the variation of them at $t=100$ is under $5.41$\text{E}-$11$.
Note in passing that the collapse here is different from that shown in \cite{ShaoTang2005}.
It was reported there that the strong negative energy and radiation appear when
the collapse happens during the binary collision of two-humped waves.
\end{Example}

\begin{Example}
\label{eg.scalar-vector}
\begin{figure}[h]
  \centering
  \subfigure[$k=1$]{
  \includegraphics[width=6.5cm]{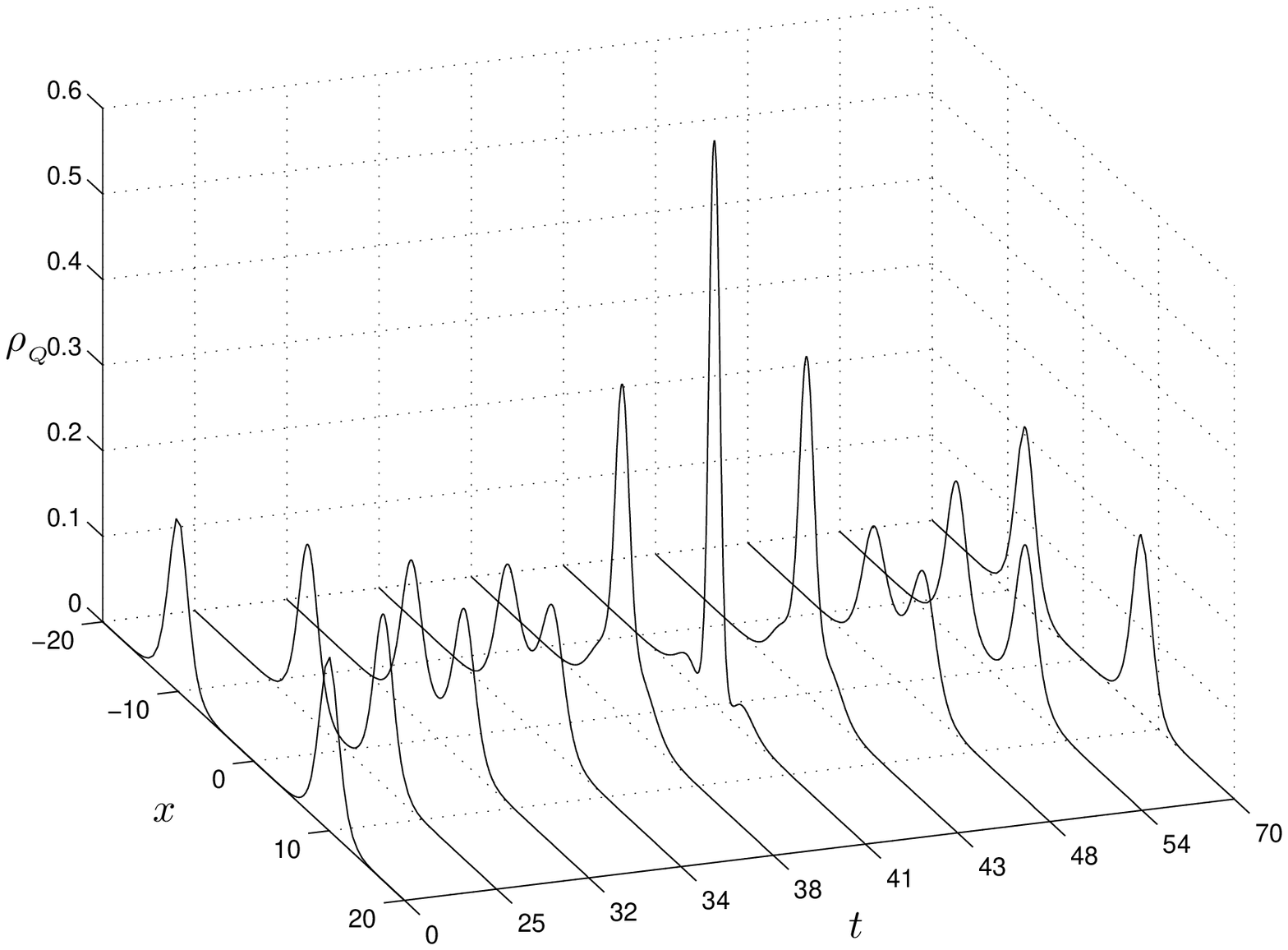}
  \label{subfigure:k1-sv-rho}}
  \subfigure[$k=2$]{
  \label{subfigure:k1-sv-QEP}
  \includegraphics[width=6.5cm]{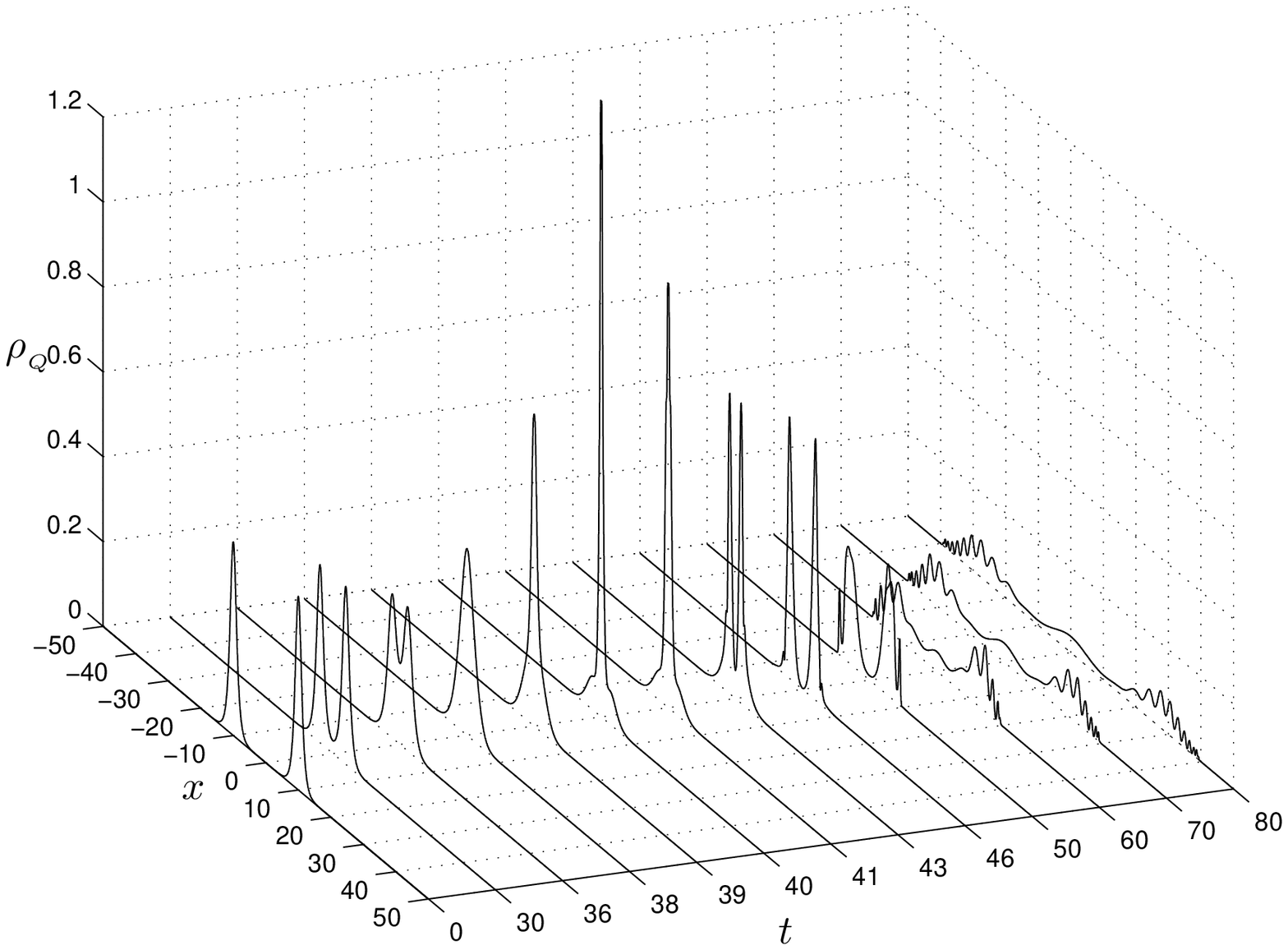}}
  \caption{Example \ref{eg.scalar-vector}: Binary collision of the NLD solitary waves under the scalar and vector self-interaction.}
  \label{fig:k1-sv}
\end{figure}

This example is devoted into investigating
for the first time the collision of two equal NLD solitary waves under the scalar and vector self-interaction,
\ie Case B3 in Table~\ref{3caselist}.
The interaction dynamics for the quadric case are shown in the left plot of Fig.~\ref{fig:k1-sv},
where two equal waves with the initial amplitude of $0.2041$ move close at a velocity of $0.2$ and overlap each other, 
then separate into a left moving wave and a right moving wave with the amplitude of $0.2091$ and the velocity of $0.1968$.
The collapse similar to that shown in right plot of Fig.~\ref{fig:k1-v} happens again for the
cubic vector self-interaction, see the right plot of Fig.~\ref{fig:k1-sv}.
The initial one-humped equal waves first merge into a single wave at $t=38$, then separate and overlap again.
Around $t=50$, collapse happens and highly oscillatory waves are generated and moving outside with a big velocity near $1$.
In both cases, the discrete charge, energy and linear momentum are approximately conserved in the interaction
since the variation of them at $t=80$ is under $3.53$\text{E}-$10$.
\end{Example}

\begin{Example}
\label{eg.scalar-2humped}

\begin{figure}[h]
  \centering
  \subfigure[$k=1$]{
  \includegraphics[width=6.5cm]{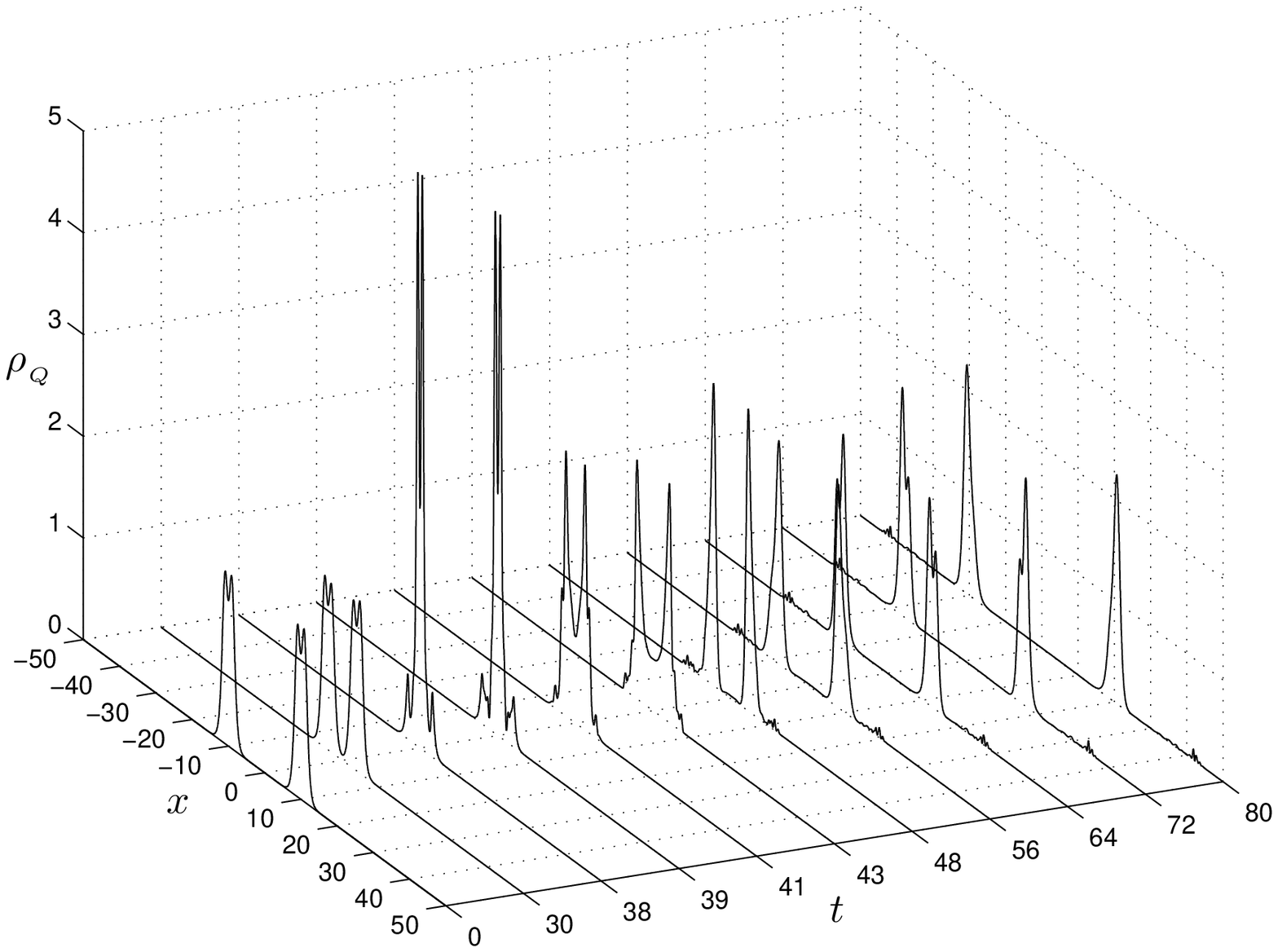}
    }
  \subfigure[$k=2$]{
    \includegraphics[width=6.5cm]{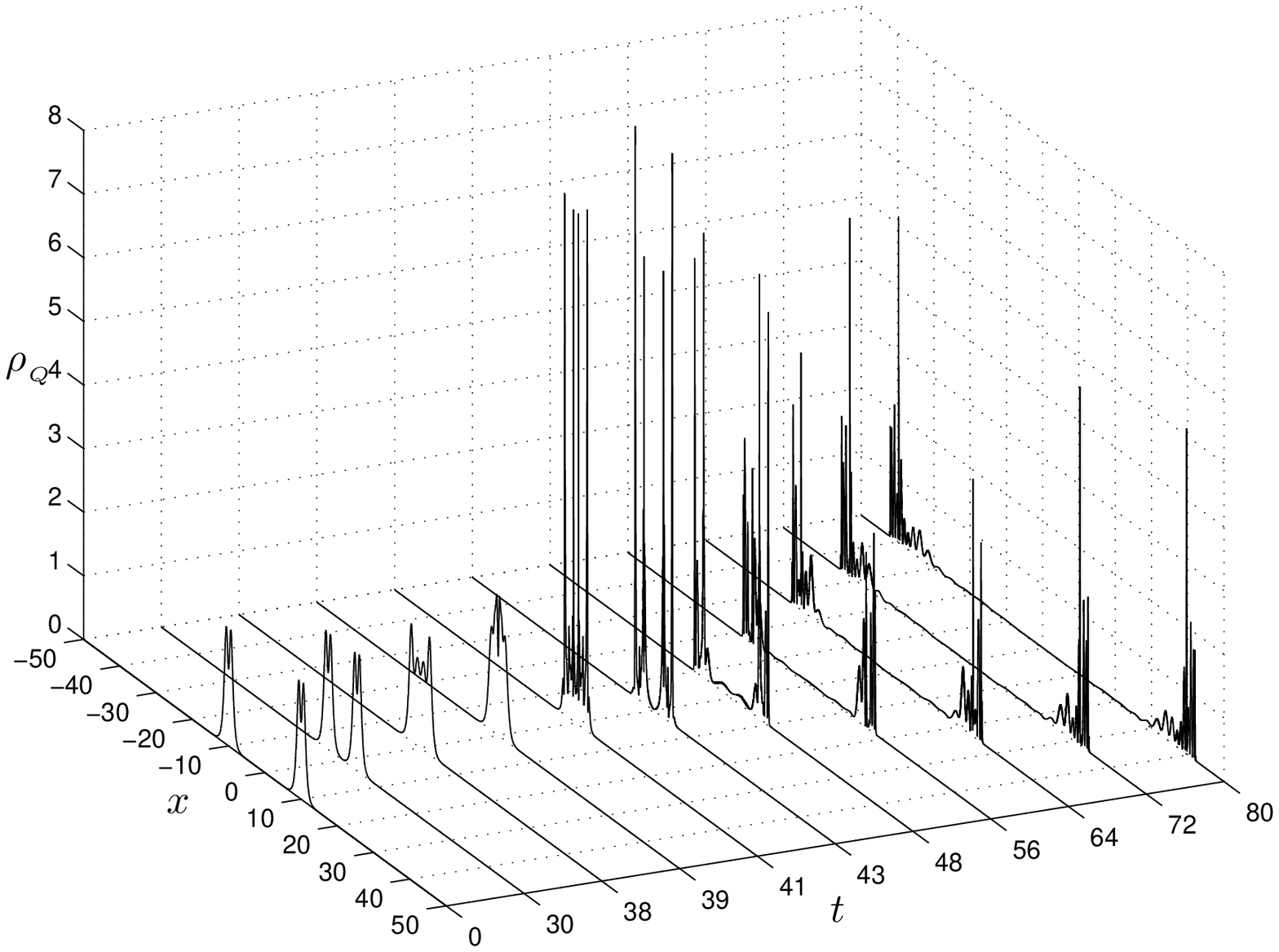}}
\caption{Example~\ref{eg.scalar-2humped}: Binary collision of the two-humped
NLD solitary waves under the scalar self-interaction.}
  \label{fig:k1-2humped}
\end{figure}

\begin{figure}[h]
  \centering
  \subfigure[$k=1$]{
  \includegraphics[width=6.5cm]{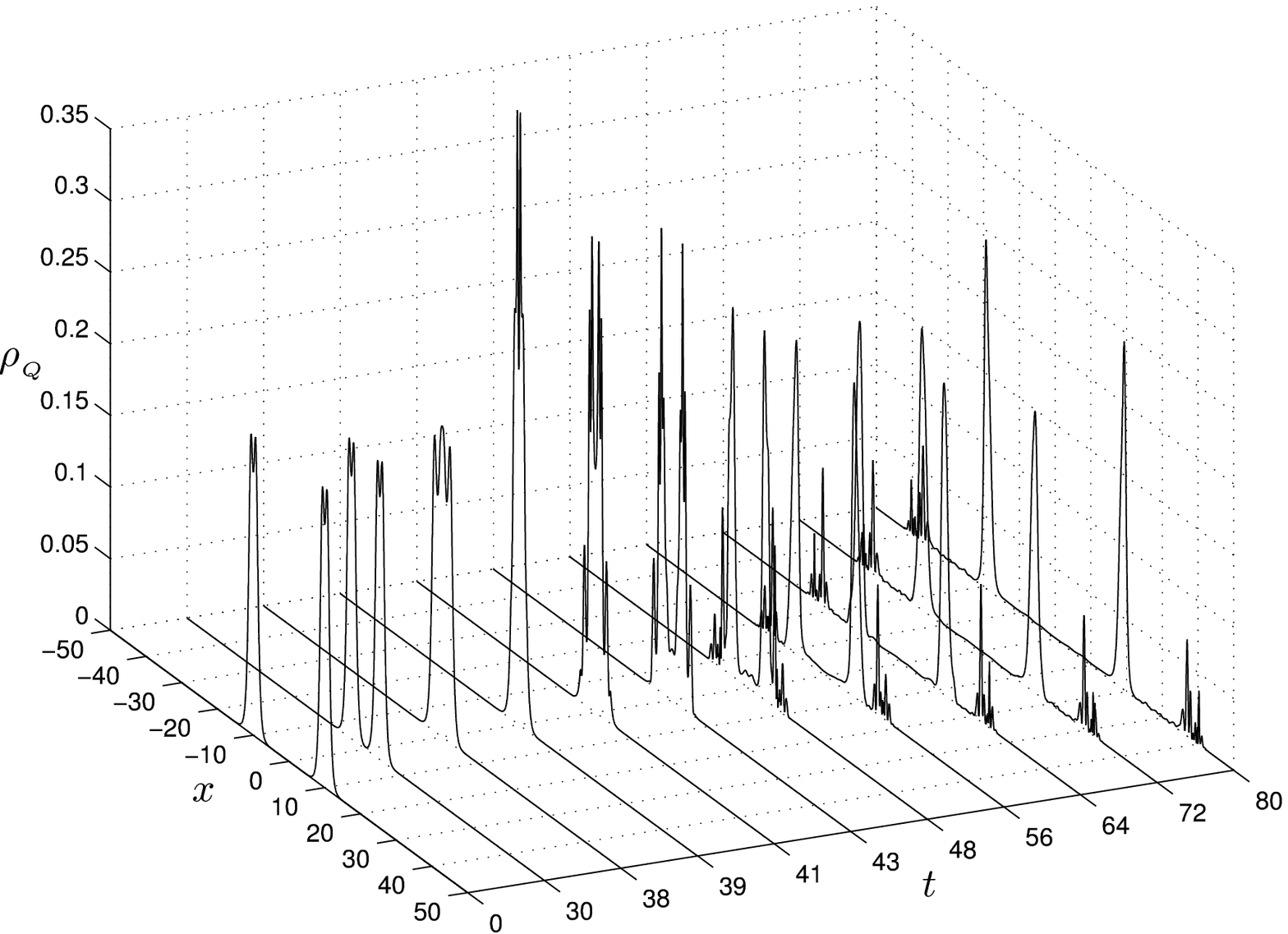}
  \label{subfigure:k1-s-rho-2humped}}
  \subfigure[$k=2$]{
  \label{subfigure:k1-s-QEP-2humped}
  \includegraphics[width=6.5cm]{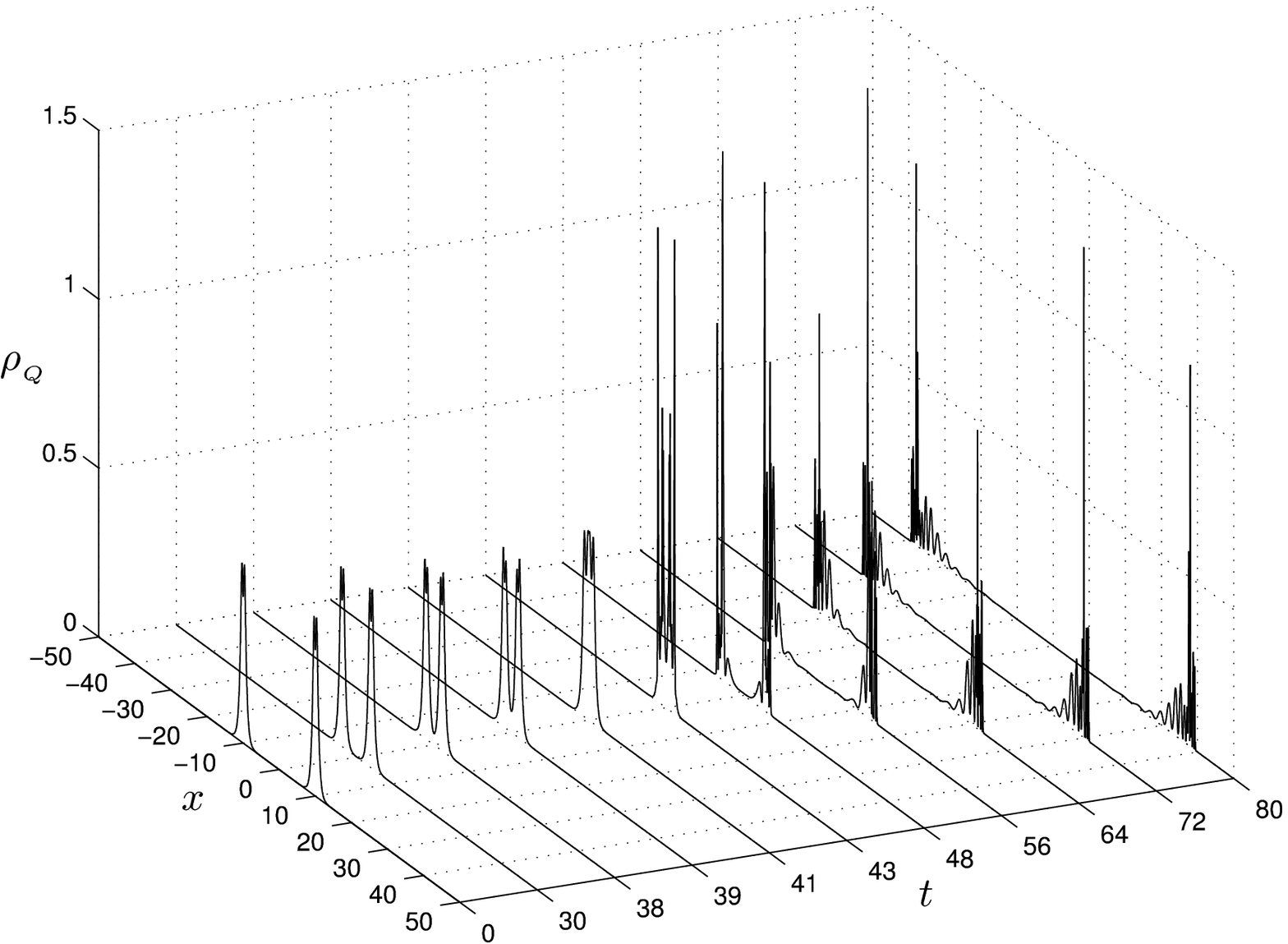}}
\caption{Example~\ref{eg.scalar-2humped}: Binary collision of the two-humped
NLD solitary waves under the scalar and vector self-interaction.}
  \label{fig:k1-2humped-sv}
\end{figure}

As reported before in \cite{ShaoTang2005,ShaoTang2008},
collapse happens in binary and ternary collisions
of the NLD solitary waves under the quadric scalar self-interaction if
the two-humped waves are evolved. In this example,
we will show further that collapse could happen in binary collision of
equal two-humped waves under the cubic scalar self-interaction and under the
linear combination of scalar and vector self-interactions.
First, Case B4 in Table~\ref{3caselist} is studied
and the interaction dynamics are shown in Fig.~\ref{fig:k1-2humped},
which clearly shows that (1) collapse happens in both quadric and cubic cases but is more stronger in the latter;
(2) two initial waves
at the same velocity are decomposed into  groups with different velocities after the collision,
but there is no such decomposition for the cubic case.
In the left plot of Fig.~\ref{fig:k1-2humped}, the highly oscillating waves with small amplitude move outside
at a big velocity of $0.9644$, while the one-humped waves with big amplitude follow them at a small velocity of $0.4626$.
In both cases, the discrete charge, energy and linear momentum are approximately conserved in the interaction
since the variation of them at $t=80$ is under $1.01$\text{E}-$5$.
Second, binary collision of equal two-humped solitary waves under the scalar and vector self-interaction,
\ie Case B5 in Table~\ref{3caselist}, is plotted in Fig.~\ref{fig:k1-2humped-sv}.
The phenomena are very similar to that shown in Fig.~\ref{fig:k1-2humped},
and the ``decomposition" phenomenon for the quadric case is more obvious
than that shown in the left plot of Fig.~\ref{fig:k1-2humped}.
\end{Example}

\section{Conclusion and outlook}
\label{sec:conclusion}

Several numerical methods for solving the NLD equation
with the scalar and vector self-interaction
have been presented and compared theoretically and numerically.
Our results have revealed that among them, the {\tt OS($4$)} scheme, one of the fourth-order accurate OS methods,
performs best in terms of the accuracy and the efficiency.
Particularly, the {\tt OS($4$)} scheme
is usually more accurate than the $P^4$-RKDG method in the mesh of the same size,
but the former needs much more less computational cost than the latter.
Such superior performance of the OS methods
is credited to the full use of the local conservation laws of the NLD equation
such that the nonlinear subproblems resulted from them are exactly solved.
The interaction dynamics for the NLD solitary waves
under the quadric and cubic self-interaction have been investigated
with the {\tt OS($4$)} scheme.
We have found that
such interaction dynamics depend on the exponent power of the self-interaction.
Actually,
it has been observed for the first time in our numerical experiments that,
(1) collapse happens in collision of two equal one-humped
NLD solitary waves under the cubic vector self-interaction
but such collapse does not appear for corresponding quadric case;
(2) two initial waves at the same velocity
are decomposed into groups with different velocities after
collapse in binary collision of two-humped NLD solitary
waves under the quadric scalar self-interaction or under the quadric scalar and vector self-interaction
but such phenomenon does not show up for corresponding cubic case.
More efforts on the interaction dynamics for the NLD solitary waves under more general self-interaction
with the {\tt OS($4$)} method are still going on.

\section*{Acknowledgments}
{Sihong Shao was partially supported by the National
Natural Science Foundation of China (Project No. 11101011)
and the Specialized Research Fund for the Doctoral Program of Higher Education (Project No. 20110001120112).
Huazhong Tang was partially supported by the National Natural Science Foundation of China (Project No. 10925101).
The authors would also like to thank the referees for many useful suggestions.}

\end{document}